\documentclass[lettersize,journal]{IEEEtran}
\usepackage{amsmath,amsfonts}
\usepackage{algorithmic}
\usepackage{array}
\usepackage{textcomp}
\usepackage{stfloats}
\usepackage{url}
\usepackage{verbatim}
\usepackage{graphicx}
\def\BibTeX{{\rm B\kern-.05em{\sc i\kern-.025em b}\kern-.08em
    T\kern-.1667em\lower.7ex\hbox{E}\kern-.125emX}}
\usepackage{balance}
\usepackage{amsmath,amsfonts}

\usepackage[colorlinks,urlcolor=blue,linkcolor=blue,citecolor=blue]{hyperref}
\usepackage{cite}

\usepackage{caption}
\usepackage{subfigure}

\usepackage{amssymb}
\usepackage{amsthm}
\usepackage{dsfont}
\usepackage{mathtools}
\usepackage{threeparttable}
\usepackage{colortbl}
\usepackage{siunitx}
\usepackage{multirow}
\usepackage{gensymb}
\usepackage{mathrsfs}
\usepackage{cases}

\DeclareSIUnit[]{\pu}{p.u.}
\DeclareSIUnit[]{\VA}{VA}

\DeclareSymbolFont{bbold}{U}{bbold}{m}{n}
\DeclareSymbolFontAlphabet{\mathbbold}{bbold}

\newcommand{\diag}[1]{\ensuremath{\mathrm{diag}(#1)}}

\DeclarePairedDelimiterX\Set[2]{\lbrace}{\rbrace}%
{ #1 \,\delimsize| \,\mathopen{} #2 }

\newcommand{\real}[0]{\mathbb R}

\usepackage{tikz}
\newcommand*\circled[1]{\tikz[baseline=(char.base)]{\node[shape=circle,draw,inner sep=0.05pt] (char) {#1};}}
\newtheoremstyle{bfnote}%
{}{}%
{\itshape}{}%
{\bfseries}{.}%
{ }%
{\thmname{#1}\thmnumber{ #2}\thmnote{ (#3)}}
\theoremstyle{bfnote}
\newtheorem{thm}{Theorem}

\newtheorem{rem}{Remark}
\newtheorem{lem}{Lemma}
\newtheorem{ass}{Assumption}
\newtheorem{definition}{Definition}
\newtheorem{prop}{Proposition}

\newtheoremstyle{bfnote}%
{}{}%
{\itshape}{}%
{\bfseries}{.}%
{ }%
{\thmname{#1}\thmnumber{ #2}\thmnote{ (#3)}}
\theoremstyle{bfnote}
\newtheorem{dyn-g}{Generator Dynamics}
\newtheorem{dyn-i}{Inverter Dynamics}

\allowdisplaybreaks

\usepackage{enumitem}
\setlist[enumerate]{leftmargin=*}
\setlist[itemize]{leftmargin=*}

\begin{document}
\title{Frequency Shaping Control for Oscillation Damping in Weakly-Connected Power Network: A Root Locus Method}

\author{Yan Jiang, Wei Chen, Zhaomin Lyu, Xunning Zhang, Dan Wang, and Shinji~Hara,~\IEEEmembership{Fellow,~IEEE}
\thanks{This work was supported in part by CUHKSZ University Development Fund, in part by National Natural Science Foundation of China under Grants 62573006 and 72131001. (\em Corresponding author: Dan Wang.)}
\thanks{Y. Jiang, Z. Lyu, and X. Zhang are with the School of Science and Engineering, The Chinese University of Hong Kong, Shenzhen, 518172, CHN. Emails: {\tt \{yjiang,zhaominlyu,xunningzhang\}@cuhk.edu.cn}.}
\thanks{W. Chen is with the School of Advanced Manufacturing and Robotics \& State Key Laboratory for Turbulence and Complex Systems, Peking University, Beijing, 100871, CHN. Email: {\tt w.chen@pku.edu.cn}.}
\thanks{D. Wang is with the School of Robotics and Automation, Nanjing University, Suzhou, 215163, CHN. Email: {\tt danwang@nju.edu.cn}.}
\thanks{S. Hara is with the Supercomputing Research Center, Institute of Integrated Research, Institute of Science Tokyo, 
Tokyo, 152-8550, JPN. Email: {\tt shinjihara5202@gmail.com}.}
}

\markboth{Journal of \LaTeX\ Class Files,~Vol.~18, No.~9, September~2020}%
{How to Use the IEEEtran \LaTeX \ Templates}

\maketitle

\begin{abstract}
Frequency control following a contingency event is of vital concern in power
system operations. Leveraging inverter-based resources, it is not hard to shape the center of inertia (COI) frequency nicely. However, under weak grid conditions, it becomes insufficient to solely shape the COI frequency since this aggregate signal fails to reveal the inter-area oscillations therein which may lead to equipment damage or extreme failures. 
To damp these oscillations, notable work typically relies on running dynamical simulations or solving optimization problems for controller design, which provides limited insight. There are some recent efforts on uncovering the role of network spectrum
and system parameters in spatiotemporal dynamics
of oscillations through eigen-structure, but oscillation damping control tuning remains nontrivial. In this manuscript, we advocate for foolproof fine-tuning rules for \emph{frequency shaping control} (FS) based on
a systematic analysis of damping ratio and decay rate of inter-area oscillations to simultaneously meet specified metrics for frequency security and oscillatory
stability. To this end, building on a modal decomposition, 
we simplify the oscillation damping problem into a pole-placement task for a set of scalar subsystems, which can be efficiently solved by only investigating the root locus of a scalar subsystem associated with the main mode, while FS inherently guarantees a Nadir-less COI frequency response. Through our proposed root-locus-based oscillatory stability analysis, we derive closed-form expressions for the minimum damping ratio and decay rate among inter-area oscillations in terms of networked system and control parameters under
FS. Specifically, this helps unveil that the minimum damping ratio is determined by the largest eigenvalue of the scaled network Laplacian. Moreover, we propose useful tuning guidelines for FS which need only simple calculations or visualized tuning to not only shape the
COI frequency into a first-order response that converges
to a steady-state value within the allowed range but also
ensure a satisfactory damping ratio and decay rate of inter-area oscillations following disturbances. As for the common virtual inertia control (VI), although similar oscillatory stability analysis becomes intractable, one can still glean some insights via the root locus method. Numerical simulations validate the proposed tuning for FS as well as the superiority of FS over VI in exponential convergence rate.
\end{abstract}


\section{Introduction}\label{sec:intro}

\IEEEPARstart{F}{requency} security following major power imbalances is a
critical concern in power system operations. Especially as conventional synchronous generators are gradually replaced by inverter-interfaced resources for lower carbon footprints, frequency becomes more volatile due to the loss of system inertia~\cite{milano2018}.    However, the fact that a power system usually has a variety of highly-coupled interaction among individual buses makes it nontrivial to analyze or tune the performance of post-contingency frequency responses. One way to reduce the size of the problem is to simply study the center of inertia (COI) frequency, which is the inertia-weighted average of the frequency signals at individual buses~\cite{milano2017rotor,azizi2020local}. The COI frequency serves as a good proxy for system-wide frequency dynamics when all buses exhibit relatively coherent behavior~\cite{Min2021lcss}. By smartly designing inverter-based control, it is not hard to shape the COI frequency dynamics following a sudden power imbalance into a first-order one that naturally has no Nadir, while at the same time keeping the steady-state frequency deviation within pre-specified limits~\cite{jiangtps2021,jiang2021lcss}. This can be achieved either by properly tuning the widely adopted virtual inertia control (VI)~\cite{arani2012implementing,fang2017distributed} or by implementing our recently proposed \emph{frequency shaping control} (FS), both of which largely enhance frequency security in low-inertia power systems with minimal algebraic calculations required for parameter tuning\cite{jiang2021tac,jiangtps2021}.

However, practical power grids often involve relatively weak tie lines connecting geographically separated groups of generation units~\cite{Kundur1991tps}. For instance, weak connectivity may arise as increased amount of renewable resources located in remote areas enter the grid through long transmission lines, such as off-shore wind farms. In such weakly-connected power networks, it is insufficient to solely rely on the COI frequency since this average signal masks the inter-area oscillations that can occur~\cite{Khamisov2024TPS,JIANG2024epsr}, which can cause minor inconvenience or catastrophic failures to system operation~\cite{entsoe2017,Kosterev1999TPS}. In fact, oscillatory stability has been explicitly mandated in many grid codes. For example, Western Power in Australia requires that oscillations be adequately damped to satisfy two criteria: any oscillation must have a damping ratio of at least $0.1$ and a halving time of at most $\SI{5}{\second}$~\cite{WP2021}. Thus, it is important to better understand how those oscillations are affected by network properties and system parameters, which will enable us to adjust inverter-based control to damp oscillations in a principled way. 


Considering the intricate spatiotemporal dynamics of oscillations in large-scale power systems, power engineers typically resort to dynamical simulations when studying the impact of network structure and parameter changes~\cite{rogers2000, Gautam2009TPS, Pagnier2019Plos, zhangtseoscillation}, which, however, is inefficient as a substantial number of scenarios must be simulated before one can gain insight. To avoid extensive simulations, the oscillation damping problem is often formulated as a robust controller design problem~\cite{Majumder2006TPS,Feng2025TPS}, which ultimately reduces to a convex optimization problem constrained by linear matrix inequalities. Unfortunately, due to the difficulty of tracking the sensitivity of the optimal solution to power network parameters, this approach provides limited insight into the impact of the network structure and parameter values on
oscillation characteristics of the system. Alternatively, many interesting patterns related to oscillations under weak grid conditions have been uncovered from the eigen-structure of the network Laplacian~\cite{guo2018cdc,Minl4dc23cluster,JIANG2024epsr} or state matrix~\cite{rouco1998eigenvalue,Ke2011TPS}, but the oscillation damping control tuning remains not straightforward.

In this manuscript, we aim to propose user-friendly tuning guidelines for FS and/or VI to ensure that a system simultaneously satisfies specified metrics for frequency security and oscillatory
stability, where only simple algebraic operations or visualized tuning procedures are needed. With this aim, we adopt a mild yet insightful proportionality assumption~\cite{pm2019preprint,jiang2021tac}, which enables us to decompose the frequency deviations after a sudden power imbalance into the COI frequency and a
transient component of inter-area oscillations~\cite{pm2019preprint,JIANG2024epsr}. This modal decomposition allows us to focus on addressing oscillatory stability while utilizing existing Nadir elimination techniques~\cite{jiang2021tac,jiangtps2021} to shape the COI frequency for frequency security. Our main contributions are as follows: 
\begin{enumerate}
    \item We turn the frequency oscillation damping problem into a pole-placement task for a set of scalar subsystems obtained from modal decomposition, which can be efficiently solved by purely analyzing the root locus of a single scalar subsystem associated with the main mode since poles of these subsystems are proven to lie neatly on this locus.
    \item We derive closed-form expressions of the minimum damping ratio and decay rate---two key oscillatory stability criteria---among inter-area oscillations for the system under FS via our proposed root locus method, which provide a quantitative characterization of how the network spectrum and parameter values influence oscillatory stability under FS. Although analogous quantitative results remain intractable for VI, the root locus under both control helps reveal that the minimum damping ratio is exclusively determined by the largest eigenvalue of the scaled power network Laplacian.
    \item We propose foolproof tuning guidelines for FS based on derived expressions to simultaneously achieve pre-defined system performance metrics for frequency security and oscillatory stability, requiring only simple algebraic operations or visualized tuning procedures. More precisely, by following our guidelines, when the system undergoes a sudden power imbalance, FS can not only shape the COI frequency into a first-order response that converges to a steady-state value within the allowed range but also ensure a satisfactory damping ratio and decay rate of inter-area oscillations. Overall, frequency deviations under FS converge exponentially fast. Typically, their convergence rate is much higher than that under VI when the same inverter inverse droop is employed.
\end{enumerate}

The rest of this manuscript is organized as follows. Section~\ref{sec:model} describes the power system model and motivates the frequency oscillation damping problem. 
Section~\ref{sec:osci-problem} simplifies the oscillation damping problem using proportionality assumption into pole-placement for a set of subsystems, whose poles are shown to reside nicely on the root locus of the scalar system associated with the main mode, enabling an efficient root-locus-based oscillatory analysis. Section~\ref{sec:FS} analyzes the oscillatory stability of the system under FS by exploiting the key root locus, upon which tuning guidelines requiring simple calculations or visualized procedures to ensure satisfactory performance are proposed for FS. Section~\ref{sec:simulation} validates our results through detailed simulations. Section~\ref{sec:conclusion} gathers our conclusions and outlook.

\section{Modeling Approach and Motivating Example}
\label{sec:model}
In this section, we describe the power system model used in
this manuscript for analysis and motivate the frequency oscillation damping problem that we
aim to solve.
\subsection{Power System Model}\label{ssec:model}
We consider\footnote{Throughout this manuscript, vectors are denoted in lower case bold and matrices are denoted in upper case bold, while scalars are unbolded, unless otherwise specified. Also, $\mathbbold{1}_n, \mathbbold{0}_n \in \real^n$ denote the vectors of all ones and all zeros, respectively, and $\boldsymbol{e}_k\in\real^n$ denotes the $k$th standard basis vector.}  a connected power network composed of $n$ buses indexed by $i \in \mathcal{N} := \{1,\dots, n\} $ and transmission lines denoted by unordered pairs $\{i,j\} \in \mathcal{E}\subset \{\{i,j\}:i,j\in\mathcal{N},i \neq j\}$, whose linearized dynamics around an operating point is shown in Fig.~\ref{fig:model}. This is a standard model and interested readers can refer to~\cite{Zhao:2013ts} for details on the linearization procedure. The system is modeled as a feedback interconnection of bus dynamics and network dynamics~\cite{pm2019preprint, jiang2021tac, jiang2021lcss}, where the input and output are the power disturbances $\boldsymbol{p} := \left(p_{i}, i \in \mathcal{N} \right) \in \real^n$ (in $\SI{}{\pu}$) and the bus frequency deviations from the nominal value $\boldsymbol{\omega}:=\left(\omega_i, i \in \mathcal{N} \right) \in \real^n$ (in $\SI{}{\pu}$), respectively. We now discuss the dynamic elements in more detail.

\begin{figure}
\centering
\includegraphics[width=\columnwidth]{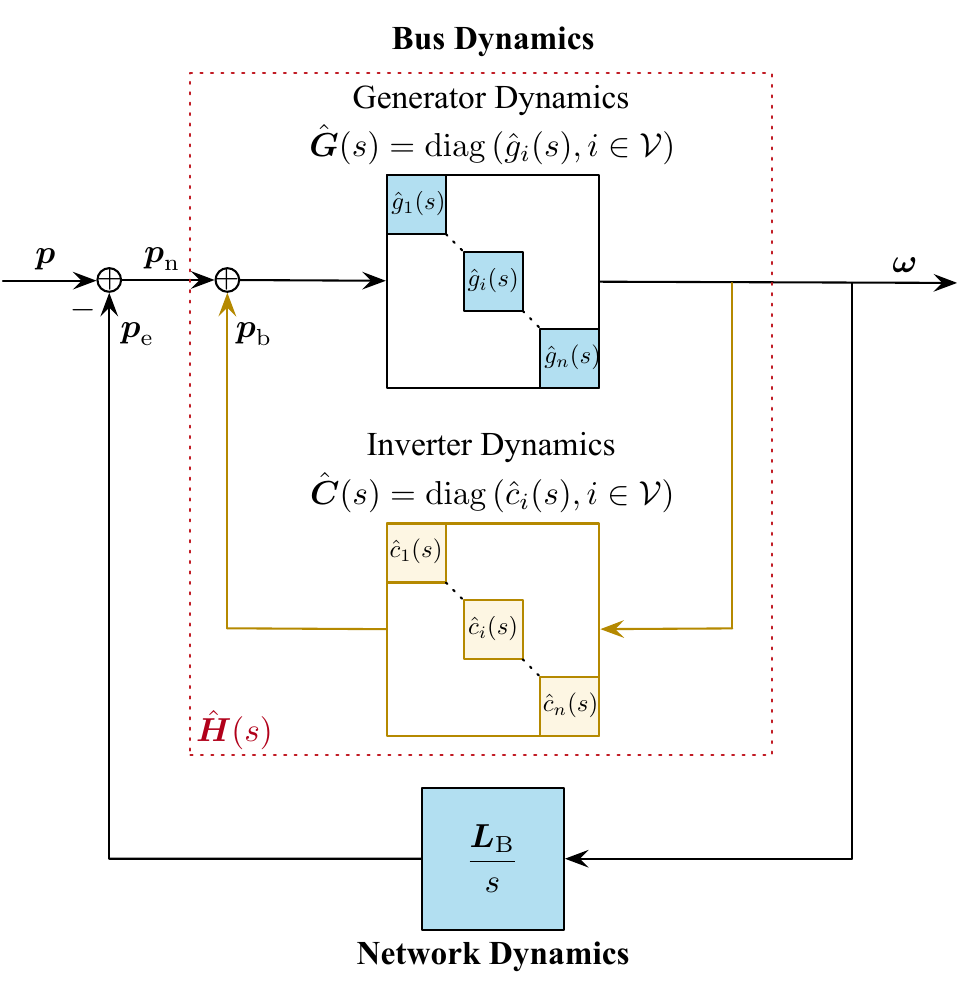}
\caption{Block diagram of power network.}\label{fig:model}
\end{figure}

\subsubsection{Bus Dynamics} The bus dynamics that maps the net power bus imbalance $\boldsymbol{p}_\mathrm{n} := \left( \boldsymbol{p}_{\mathrm{n},i}, i \in \mathcal{N} \right) \in \real^n$ to frequency deviations $\boldsymbol{\omega}$ can be described as a feedback loop that comprises a forward-path $\hat{\boldsymbol{G}}(s)$ and a feedback-path $\hat{\boldsymbol{C}}(s)$, where $\hat{\boldsymbol{G}}(s) := \diag {\hat{g}_i(s), i \in \mathcal{N}}$ and $\hat{\boldsymbol{C}}(s) := \diag {\hat{c}_i(s), i \in \mathcal{N}} $ are the transfer function matrices of generators and inverters, respectively. Thus, the transfer function from $\boldsymbol{p}_\mathrm{n}$ to $\boldsymbol{\omega}$ is 
\begin{equation}\label{eq:bus-dyn}
 \hat{\boldsymbol{H}}(s):=\left(\boldsymbol{I}_n-\hat{\boldsymbol{G}}(s)\hat{\boldsymbol{C}}(s)\right)^{-1}\hat{\boldsymbol{G}}(s)  \,. 
\end{equation}
\paragraph{Generator Dynamics} Typically, the generator dynamics is composed of the standard swing equation with a turbine, which can be described by 
\begin{equation} \label{eq:dy-sw-t}
\hat{g}_i(s) := \frac{ \tau_i s + 1 }{m_i \tau_i s^2 + \left(m_i + d_i \tau_i \right) s + d_i + d_{\mathrm{t},i}}\,,
\end{equation}
where $m_i>0$ denotes the aggregate generator inertia, $d_i>0$ the aggregate generator damping, $\tau_i>0$ the turbine time constant, and $d_{\mathrm{t},i}>0$ the turbine inverse droop.
\paragraph{Inverter Dynamics} We consider the most commonly used type of inverters, i.e., the grid-following inverters, which measure the local grid frequency deviation $\omega_i$ and instantaneously adjust their power injection variation $p_{\mathrm{b},i}$ to the network according to a control law $\hat{c}_i(s)$. For example, a straightforward approach is to use $\hat{c}_i(s)$ to compensate for the lack of physical inertia from synchronous generators.
\subsubsection{Network Dynamics}
The network dynamics characterizes the relationship between the fluctuations in power drained into the transmission
network $\boldsymbol{p}_\mathrm{e} := \left(p_{\mathrm{e},i}, i \in \mathcal{N} \right) \in \real^n$ (in $\SI{}{\pu}$) and the frequency deviations $\boldsymbol{\omega}$, which is given by a linearized model of the power flow equations~\cite{Purchala2005dc-flow}:
\begin{equation}
 \hat {\boldsymbol{p}}_\mathrm{e}(s) = \frac{\boldsymbol{L}_\mathrm{B}}{s} \hat {\boldsymbol{\omega}}(s)\;,\footnote{We use hat to distinguish the Laplace transform from its time domain counterpart.}\label{eq:Network-dyn}
\end{equation}
 where the matrix $\boldsymbol{L}_\mathrm{B}:=\left[ L_{\mathrm{B},{ij}}\right]\in \real^{n \times n}$ is an undirected weighted Laplacian matrix of the network whose $ij$th element is 
\[
\boldsymbol{L}_{\mathrm{B},{ij}}=\Omega_0\partial_{\theta_j}{\sum_{l=1}^n|V_i||V_l|B_{il}\sin(\theta_i-\theta_l)}\Bigr|_{\boldsymbol{\theta}=\boldsymbol{\theta}_0}.
\]
Here, $\boldsymbol{\theta} := \left(\theta_i, i \in \mathcal{N} \right) \in \real^n$ are the voltage angles with $\boldsymbol{\theta}_0$ being the equilibrium angles (in $\SI{}{\radian}$), $|V_i|$ is the (constant) voltage magnitude at bus $i$ (in $\SI{}{\pu}$), $B_{ij}$ is the line $\{i,j\}$ susceptance (in $\SI{}{\pu}$), and $\Omega_0:=2\pi F_0$ is the nominal frequency (in $\SI{}{\radian/s}$) with $F_0$ being $\SI{50}{\hertz}$ or $\SI{60}{\hertz}$
depending on the particular system. 

\subsubsection{Closed-Loop Dynamics} Since too large frequency deviations could trigger undesired protection measures and
even cause cascading failures, we care about the closed-loop response of frequency deviations $\boldsymbol{\omega}$ following power disturbances $\boldsymbol{p}$ in the power network shown in Fig.~\ref{fig:model}. This can be obtained by combining bus dynamics \eqref{eq:bus-dyn} and network dynamics \eqref{eq:Network-dyn} through the relation $\boldsymbol{p}_\mathrm{n}=\boldsymbol{p}-\boldsymbol{p}_\mathrm{e}$ as 
\begin{equation}\label{eq:twp}
\hat {\boldsymbol{\omega}}(s)= \left(\boldsymbol{I}_n +\hat{\boldsymbol{H}}(s)\frac{\boldsymbol{L}_\mathrm{B}}{s}\right)^{-1} \hat{\boldsymbol{H}}(s)\hat{\boldsymbol{p}}(s)\,.
\end{equation}
However, although \eqref{eq:twp} is a closed-form expression, it is difficult to work with since the size of the matrices and vectors involved could be quite high.

A standard simplification is to work with the center of inertia (COI) frequency, defined as the inertia-weighted average
of individual bus frequencies, i.e.,
\begin{equation}\label{eq:COI}
    \bar{\omega}:= \dfrac{\sum_{i=1}^n \check{m}_i\omega_i}{\sum_{i=1}^n \check{m}_i}\,,
\end{equation}
where $\check{m}_i$ denotes the total inertia on bus $i$, including both the generator inertia $m_i$ and any compensation from the inverter on bus $i$.  The COI frequency is a good representative of the system frequency response if the power network is tightly connected~\cite{Min2021lcss}. With this
underlying assumption, our previous work~\cite{jiang2021tac,jiang2021lcss,jiangtps2021} focus on leveraging inverter-based control to shape the COI response nicely for frequency security. In particular, we have provided a systematic tuning recommendation to remove the COI frequency Nadir via the widely adopted VI and our proposed FS, which are both reviewed in Section \ref{sec:FS}.

However, as more renewable resources are integrated, especially at edges of the power network, long transmission lines are often needed to bring this power to load centers, potentially contributing to inter-area oscillations at a low frequency of $0.1$--$\SI{0.8}{\hertz}$~\cite{Kundur1991tps, Khamisov2024TPS}. In fact, the frequency trajectories following sudden power disturbances may exhibit pronounced oscillations against each other even if the COI frequency is well-behaved~\cite{JIANG2024epsr}. Such oscillations can cause equipment damage and, in extreme cases, catastrophic failures. Therefore, it becomes insufficient to only damp out the COI frequency oscillation by shaping the frequency into a first-order dynamics~\cite{jiang2021tac,jiang2021lcss,jiangtps2021}. 
\subsection{Illustrative Example of Frequency Oscillations}\label{ssec:example}
To see this, we adopt the modified Western System Coordinating Council (WSCC) $9$-bus $3$-generator system in Fig.~\ref{fig:WCSS_test_case} as an example to show different types of oscillatory behaviors. The detailed parameters of this system are provided in Section~\ref{sec:simulation}. Here, we compare the frequency responses of the system with strong and weak network connectivity by varying the line impedance parameters on the same network topology, where weak connectivity is emulated by increasing the impedance of the lines $4-9$ and $5-6$ to $20$ times its original value that provides strong connectivity, resulting in a $2$-area system. 

\begin{figure}[t!]
\centering
\includegraphics[width=0.8\columnwidth]{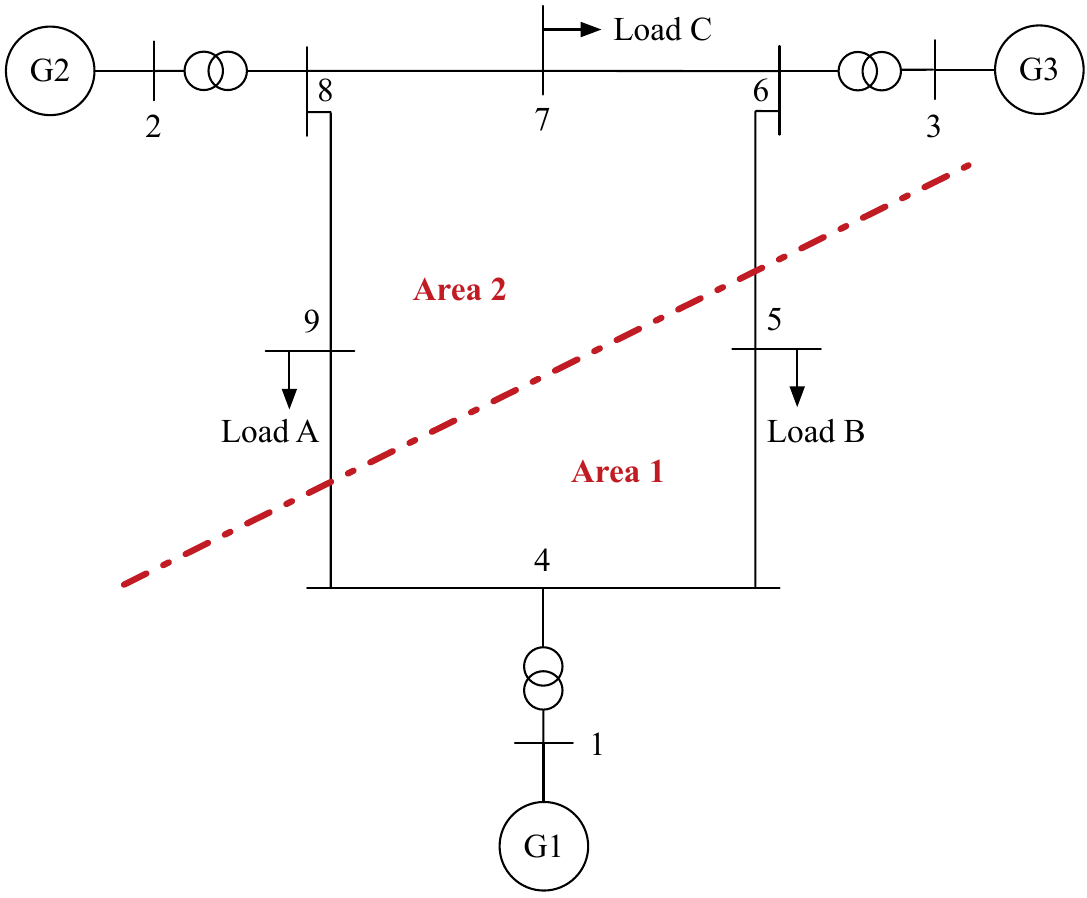}
\caption{Diagram of $9$-bus $3$-generator WSCC test case.}
\label{fig:WCSS_test_case}
\end{figure}

To briefly review the idea of shaping the COI frequency~\cite{jiang2021tac,jiang2021lcss,jiangtps2021}, the frequency deviations of the system without additional control from inverters when a sudden step power disturbance $\boldsymbol{p}= [-0.2, 0, 0]^T \mathds{U}_{ t \geq 0 }$~$\SI{}{\pu}$ occur at $t=\SI{1}{\second}$ are provided in Fig.~\ref{fig:fre-SG}.\footnote{$\mathds{U}_{ t \geq 0 } $ denotes the unit-step function.} Obviously, all buses exhibit a common oscillation triggered by the disturbance, which can be captured well by the COI frequency, independent of network connectivity. To address this system-wide oscillation that leads to a deep frequency Nadir, we then place an inverter near each generator bus, where the active power controller $\hat{c}_i(s)$ on individual buses are all given by VI or FS with Nadir elimination tuning proposed by our previous work~\cite{jiang2021tac,jiang2021lcss,jiangtps2021}.  In Fig.~\ref{fig:3-bus},
we show how these two control strategies perform when the system experiences the same power disturbance as in the previous test for both connectivity scenarios. As expected, both strategies succeed in shaping the COI frequency into a first-order response that naturally has no Nadir, which greatly improves the frequency security.

\begin{figure}[t!]
\centering
\subfigure[All buses are tightly-connected]
{\includegraphics[width=0.49\columnwidth]{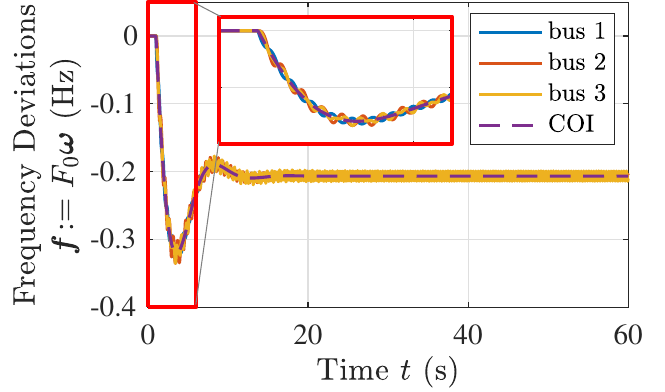}\label{fig:fre-SG-tight}}
\subfigure[Two areas are weakly-connected]
{\includegraphics[width=0.49\columnwidth]{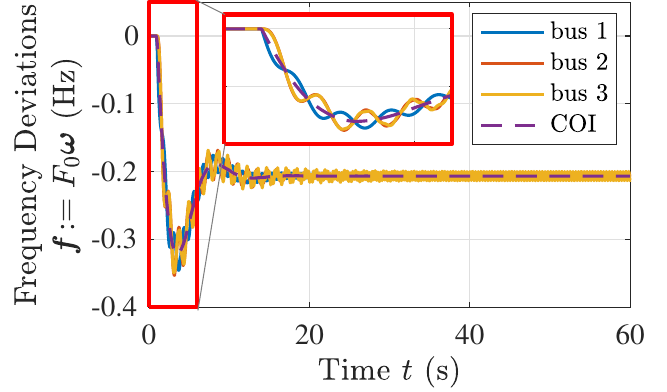}\label{fig:fre-SG-weak}}
\caption{Frequency deviations of the WSCC system without inverters when a $\SI{-0.2}{\pu}$ step power change is introduced to bus $1$.}\label{fig:fre-SG}
\end{figure}

However, unlike the more tightly-connected system in Fig.~\ref{fig:fre-SG-tight} and Fig.~\ref{fig:fre-VIFS-notune-tight}, where the frequency trajectories at all buses closely track each other even though the disturbance is only at bus $1$, the COI frequency is no longer a good enough proxy of the nodal frequencies in Fig.~\ref{fig:fre-SG-weak} and Fig.~\ref{fig:fre-VIFS-notune-weak}. The problem is that, the frequency trajectories in Fig.~\ref{fig:fre-SG-weak} and Fig.~\ref{fig:fre-VIFS-notune-weak} exhibit noticeable oscillations relative to each other, while the well-behaved COI frequency fails to reveal this phenomenon. Therefore, relying just on the COI frequency could lead to an erroneous conclusion about the oscillatory stability~\cite{WP2021} of the system. Of course, the inter-area oscillations in Fig.~\ref{fig:fre-SG-weak} and Fig.~\ref{fig:fre-VIFS-notune-weak} are not unexpected since bus $1$ is ``weakly'' connected with the other buses. Nevertheless, as the power network gets larger, it is more challenging to draw intuitive conclusions, which makes additional analysis techniques become necessary.

\begin{figure}[t!]
\centering
\subfigure[All buses are tightly-connected]
{\includegraphics[width=0.98\columnwidth]{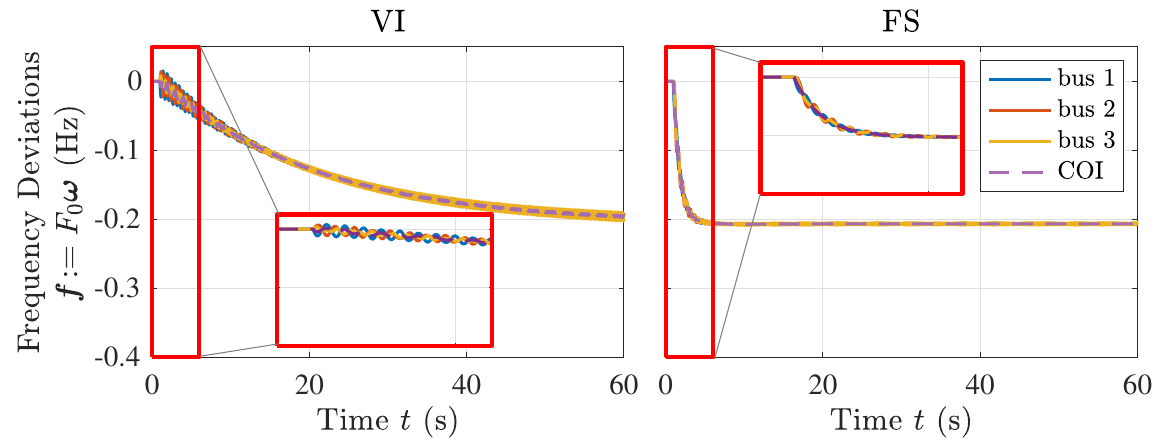}\label{fig:fre-VIFS-notune-tight}}
\subfigure[Two areas are weakly-connected]
{\includegraphics[width=0.98\columnwidth]{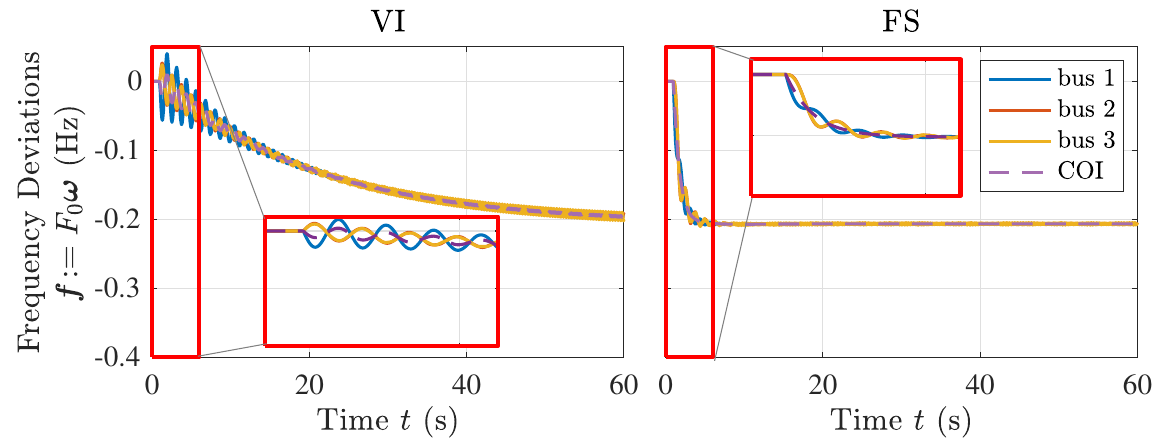}
\label{fig:fre-VIFS-notune-weak}}
\caption{Frequency deviations of the WSCC system under VI and FS control with Nadir elimination tuning when a $\SI{-0.2}{\pu}$ step power change is introduced to bus $1$.}
\label{fig:3-bus}
\vspace{-0.4cm}
\end{figure}   

Worse still, inter-area oscillations persist in the weak connectivity scenario even after inverter-based control is applied. This is unsurprising since VI and FS are tuned solely to shape the COI frequency here. Particularly, they do not provide additional droop, given that the steady-state frequency deviations of the original system without inverter-based control are already within the allowed range of~\SI{\pm200}{\milli\hertz} \cite{UCTLbook,NG2016standard}. Clearly, care must be taken when tuning VI and FS to address those oscillations. This gives rise to a natural question: \emph{is there a systematic way for fine-tuning VI and FS to not only shape the COI frequency well but also damp other oscillations?} To this end, we need to better understand how these oscillations are impacted by the network structure and parameter
values.
\section{Frequency Oscillation Damping Problem Formulation}\label{sec:osci-problem}

In this section, we show that a simplifying assumption allows us to transform the frequency oscillation damping problem into a pole-placement task for a set of scalar subsystems obtained from modal decomposition, which can be solved by purely analyzing the root locus of a scalar subsystem associated with the main mode.  

\subsection{Oscillation Damping Problem under Modal Decomposition}

As mentioned above, although our previous work~\cite{jiang2021tac,jiang2021lcss,jiangtps2021} can shape the COI frequency response to remove its Nadir, the influence of oscillations on top of COI frequency has been ignored, which could lead to an inaccurate conclusion about the oscillatory
stability of the system. Thus, it is important to damp out those oscillations. However, since the oscillations present a variety of highly-coupled interactions between individual buses, it is challenging to investigate the oscillatory behavior in large-scale power systems with heterogeneous devices. 
To make the analysis tractable, we consider proportionality as a reasonable
first-cut approximation to heterogeneity~\cite{pm2019preprint}, under which the frequency dynamics \eqref{eq:twp} are decoupled and thus the oscillation damping problem can be solved efficiently. Hence, we adopt the following assumption in the rest of this manuscript:
\begin{ass}[Proportionality]\label{ass:proportion}
There exists a group of proportional parameters $r_i>0$, $i \in \mathcal{N}$, such that 
\begin{align*} \label{eq:propotionmodel}
\hat{g}_i(s) = \dfrac{\hat{g}_\mathrm{o}(s)}{r_i} \qquad \text{and} \qquad \hat{c}_i(s) = r_i \hat{c}_\mathrm{o}(s)\,, 
\end{align*}
where $\hat{g}_\mathrm{o}(s)$ and $\hat{c}_\mathrm{o}(s)$ are called the representative generator and the representative inverter, respectively.
\end{ass}
For example, for $\hat{g}_i(s)$ in \eqref{eq:dy-sw-t}, Assumption~\ref{ass:proportion} holds if, $\forall i \in \mathcal{N}$, $\exists\, r_i>0$ such that $m_i=r_im$, $d_i=r_id$, $d_{\mathrm{t},i}=r_i d_{\mathrm{t}}$, and $\tau_i=\tau$ for some  representative parameters $m$, $d$, $d_{\mathrm{t}}$, and $\tau$, i.e., 
\begin{equation}\label{eq:go}
\hat{g}_\mathrm{o}(s)=\frac{ \tau s + 1 }{m \tau s^2 + \left(m + d \tau \right) s + d + d_{\mathrm{t}}} \,.   
\end{equation}

Under Assumption~\ref{ass:proportion}, it has been well-established that the dynamics in \eqref{eq:twp} can be decoupled as~\cite{pm2019preprint, jiang2021tac,JIANG2024epsr}
\begin{align}
\hat {\boldsymbol{\omega}}(s) = \boldsymbol{R}^{-\frac{1}{2}} \boldsymbol{V} \diag{\hat{z}_{k}(s), k \in \mathcal{N}} \boldsymbol{V}^T \boldsymbol{R}^{-\frac{1}{2}}\hat{\boldsymbol{p}}(s)\;,\label{eq:Tp}
\end{align}
with
\begin{align}
    \hat{z}_{k}(s) := \frac{\hat{g}_\mathrm{o}(s)}{1+\hat{g}_\mathrm{o}(s)\left(\lambda_k/s-\hat{c}_\mathrm{o}(s)\right)}\;,\quad \forall k \in \mathcal{N}\,, \label{eq:zk-s}
\end{align}
where $\boldsymbol{R} :=\diag {r_i, i \in \mathcal{N}} \in \real^{n \times n}$ is the proportionality matrix and\begin{align}\label{eq:Vmatrix}
\boldsymbol{V} \!:=\! \begin{bmatrix} \boldsymbol{v}_1\!:=\!(\sum_{i=1}^n \!r_i)^{-\frac{1}{2}} \boldsymbol{R}^{\frac{1}{2}} \mathbbold{1}_n \!\!& \!\boldsymbol{v}_2 &\!\cdots \!& \boldsymbol{v}_n \end{bmatrix}\!\!\in\! \real^{n \times n}    
\end{align}
satisfying $\boldsymbol{V}^T\boldsymbol{V}=\boldsymbol{V}\boldsymbol{V}^T=\boldsymbol{I}_n $ is an orthonormal matrix whose columns $\boldsymbol{v}_k:=\left(v_{k,i}, i \in \mathcal{N} \right) \in \real^n$ are unit eigenvectors associated with the scaled Laplacian matrix 
\begin{align}\label{eq:L-def}
\boldsymbol{L} := \boldsymbol{R}^{-\frac{1}{2}} \boldsymbol{L}_\mathrm{B} \boldsymbol{R}^{-\frac{1}{2}}\in \real^{n \times n}    
\end{align} 
such that $\boldsymbol{L} = \boldsymbol{V} \diag{\lambda_k, k \in \mathcal{N}} \boldsymbol{V}^T$
with $\lambda_k$ being the $k$th eigenvalue of $\boldsymbol{L}$ ordered non-decreasingly $(0 = \lambda_1 < \lambda_2 \leq \ldots \leq\lambda_n)$\footnote{Recall that we assume the power network is connected, which implies that $\boldsymbol{L}$ has a single zero eigenvalue.}.

To make the physical interpretation of modal decomposition more clear, we rewrite \eqref{eq:Tp} as
\begin{align}
    \!\!\! \!\hat {\boldsymbol{\omega}}(s) =\dfrac{\mathbbold{1}_n^T\hat{\boldsymbol{p}}(s)}{\sum_{i=1}^n \!r_i}\hat{z}_{1}(s)\mathbbold{1}_n+\!\underbrace{\sum_{k=2}^n \!\hat{z}_{k}(s)\boldsymbol{R}^{-\frac{1}{2}} 
    \boldsymbol{v}_k\boldsymbol{v}_k^T\boldsymbol{R}^{-\frac{1}{2}}\hat{\boldsymbol{p}}(s)}_{=:\tilde {\boldsymbol{\omega}}(s)}\label{eq:omega-t-decomple-coi}
\end{align} 
by explicitly substituting the expression of $\boldsymbol{v}_1$ in \eqref{eq:Vmatrix} to it.
Here, the first mode actually characterizes the common behavior among individual buses as the representative feedback loop $\hat{z}_{1}(s)$, responding to the total disturbance $\mathbbold{1}_n^T\hat{\boldsymbol{p}}(s)$ scaled by $1/(\sum_{i=1}^n r_i)$, while the remaining modes represent the oscillations among them. In fact, one can easily show by a similar argument as in \cite{pm2019preprint} that the common behavior is the $r_i$-weighted average of individual bus frequencies, i.e., $(\sum_{i=1}^n r_i\omega_i)/(\sum_{i=1}^n r_i)$, which is exactly the COI frequency defined in \eqref{eq:COI}. Therefore, the modal decomposition in \eqref{eq:omega-t-decomple-coi} confirms that it becomes insufficient to only design controller based on COI frequency when those oscillatory modes are unignorable.

Therefore, our goal is to design $\hat{c}_\mathrm{o}(s)$ to not only remove Nadir of the COI frequency but also damp oscillations in $\tilde {\boldsymbol{\omega}}(s)$. Since the former has been solved by~\cite{jiang2021tac,jiang2021lcss,jiangtps2021}, the latter will be the core of this manuscript. With this in mind, we first note from \eqref{eq:omega-t-decomple-coi} that the transfer function matrix from $\hat{\boldsymbol{p}}(s)$ to $\tilde {\boldsymbol{\omega}}(s)$, i.e.,
\begin{align*}
    \hat{\boldsymbol{T}}_{\tilde {\boldsymbol{\omega}}\boldsymbol{p}}(s)=\sum_{k=2}^n \hat{z}_{k}(s)\boldsymbol{R}^{-\frac{1}{2}} 
    \boldsymbol{v}_k\boldsymbol{v}_k^T\boldsymbol{R}^{-\frac{1}{2}}\,,
\end{align*}
has all elements as a linear combination of $\hat{z}_{k}(s)$ for $k\in\mathcal{N}\setminus\{1\}$. Thus, if we can constrain poles of each $\hat{z}_{k}(s)$ to lie in a prescribed region such that a certain damping ratio and decay rate can be achieved~\cite{Hara2014tac}, we can ensure a satisfactory transient response of $\hat{\boldsymbol{T}}_{\tilde {\boldsymbol{\omega}}\boldsymbol{p}}$ with oscillations adequately damped as required by the grid code~\cite{WP2021}. Therefore, the oscillation damping problem boils down to a pole-placement problem.
\subsection{Pole-Placement Problem Simplification via Root Locus}\label{ssec:locus-method}
Now, given any pre-designed $\hat{c}_\mathrm{o}(s)$ that is engineered to eliminate the COI frequency Nadir as will be detailed in Sections~\ref{sec:FS}, our main objective is to prune its tuning region so that all poles of $\hat{z}_{k}(s)$ for $k\in\mathcal{N}\setminus\{1\}$ lie within a region $\mathcal{S}_{\alpha,\psi}$ as illustrated in Fig.~\ref{fig:S-region}, which ensures a minimum decay rate $\alpha>0$ and a minimum damping ratio $\cos\psi$ for $\psi\in[0,\pi/2)$. This can be considered as specifying stability degrees through the damping ratio and decay rate rather than merely guaranteeing stability of those subsystems. Here, it is natural to define $(\alpha, \psi)$-stability:

\begin{definition}[$(\alpha, \psi)$-stability] A system is called $(\alpha, \psi)$-stable if all its poles lie in a region $\mathcal{S}_{\alpha,\psi}$ depicted by Fig.~\ref{fig:S-region}. 
\end{definition}

Thus, damping oscillations in $\tilde {\boldsymbol{\omega}}(s)$ is simply a matter of securing the specified $(\alpha, \psi)$-stability of $\hat{z}_{k}(s)$ for $k\in\mathcal{N}\setminus\{1\}$ through pole-placement. This type of pole-placement problem has gained attention from both control and power communities, yet most existing studies attempt to formulate it as a convex optimization problem involving linear matrix inequalities~\cite{Chilali1996tac,Feng2025TPS}, which provides limited insight into the precise impact of network structure and parameter values on oscillation characteristics of the system. To overcome this, we aim to devise a principled procedure for pole-placement by fine-tuning $\hat{c}_\mathrm{o}(s)$, where closed-form expressions of $\cos\psi$ and $\alpha$ in terms of the system and control parameters can be provided.

\begin{figure}[t!]
\centering
\includegraphics[width=0.35\columnwidth]{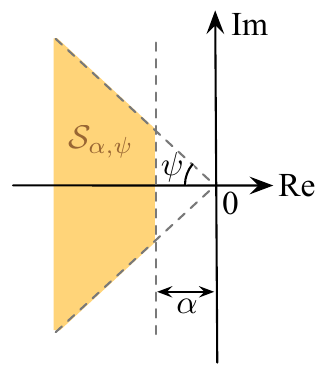}
\caption{Region $\mathcal{S}_{\alpha,\psi}$.}\label{fig:S-region}
\end{figure}

Clearly, a better understanding of how pole positions of $\hat{z}_{k}(s)$ vary as $k$ increases under given $\hat{c}_\mathrm{o}(s)$ is crucial to achieving this objective. However, instead of directly analyzing the pole locations of each individual $\hat{z}_{k}(s)$ in \eqref{eq:zk-s}, we adopt a more efficient approach: we only examine the root locus of the negative feedback interconnection of $\hat{z}_{1}(s)/s$ and a scalar gain $\lambda$ depicted in Fig.~\ref{fig:rootlocussys}, i.e.,
\begin{align}\label{eqfeed-lambda}
 \dfrac{\hat{z}_{1}(s)/s}{1+\lambda\hat{z}_{1}(s)/s}=:\mathcal{F}(\hat{z}_{1}(s)/s,\lambda) \,,
\end{align}
as $\lambda$ varies, which is the core idea that underpins our subsequent analysis in Sections~\ref{sec:FS}. This is captured by the following theorem.

\begin{thm}[Poles of $\hat{z}_{k}(s)$ residing on root locus of $\mathcal{F}(\hat{z}_{1}(s)/s,\lambda)$]\label{thm:colocate-z}
For the system in Fig.~\ref{fig:model} under Assumption~\ref{ass:proportion}, if $\hat{c}_\mathrm{o}(s)$ is designed to make the root locus of $\mathcal{F}(\hat{z}_{1}(s)/s,\lambda)$ lie in the open left half-plane for any $\lambda>0$, then, $\forall k\in\mathcal{N}\setminus\{1\}$, the poles of $\hat{z}_{k}(s)$ are exactly the points that yield the gain $\lambda=\lambda_k$ on the root locus of the system with the open-loop transfer function  
\begin{align}\label{eq:Loopgain}
    \hat{L}(s):= \dfrac{\lambda\hat{z}_{1}(s)}{s}\,,
\end{align}
where $\lambda$ plays the role of variable gain in the normal case. 
\end{thm}
\begin{proof}
\begin{figure}[t]
\centering
\includegraphics[width=0.35\columnwidth]{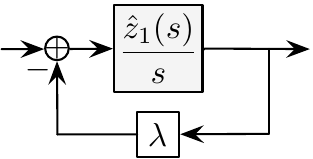}
\caption{Negative feedback
interconnection of $\hat{z}_{1}(s)/s$ and $\lambda$, or $\mathcal{F}(\hat{z}_{1}(s)/s,\lambda)$ for short.}\label{fig:rootlocussys}
\end{figure}

To see why this is true, we first note that the block diagram of $\hat{z}_{k}(s)$ in Fig.~\ref{fig:zk(s)} can be equivalently transformed into the one in Fig.~\ref{fig:zk-v2} by moving the pickoff point past the integrator~\cite[Figure 2.7]{Driels1996linear}. This transformation reveals that $\hat{z}_{k}(s)$ only introduces an additional zero at the origin to the negative feedback interconnection of $\hat{z}_{1}(s)/s$ and $\lambda_k$, or $\mathcal{F}(\hat{z}_{1}(s)/s,\lambda_k)$ for short. Thus, for any given $k\in\mathcal{N}\setminus\{1\}$, as long as we design $\hat{c}_\mathrm{o}(s)$ to make all poles of $\mathcal{F}(\hat{z}_{1}(s)/s,\lambda_k)$ lie strictly in the open left half-plane---thereby ensuring no pole at the origin, then we know $\hat{z}_{k}(s)$ only adds a zero at the origin that cannot be canceled by any pole to $\mathcal{F}(\hat{z}_{1}(s)/s,\lambda_k)$. We will illustrate soon that the requirement on $\hat{c}_\mathrm{o}(s)$ above is guaranteed by our assumption on $\hat{c}_\mathrm{o}(s)$. Therefore, the pole set of $\mathcal{F}(\hat{z}_{1}(s)/s,\lambda_k)$ is identical to that of $\hat{z}_{k}(s)$, which allows us to focus only on the pole-placement of $\mathcal{F}(\hat{z}_{1}(s)/s,\lambda_k)$ from now on.

\begin{figure}[t!]
\centering
\subfigure[Block diagram of $\hat{z}_{k}(s)$ in \eqref{eq:zk-s}]
{\includegraphics[width=0.5\columnwidth]{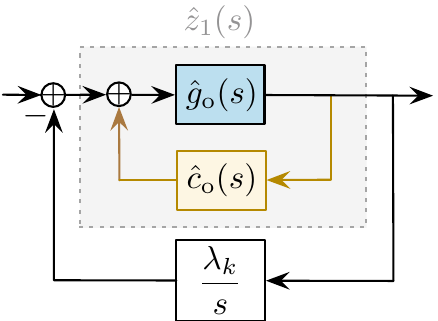}\label{fig:zk(s)}}
\subfigure[Equivalent block diagram of $\hat{z}_{k}(s)$]
{\includegraphics[width=0.65\columnwidth]{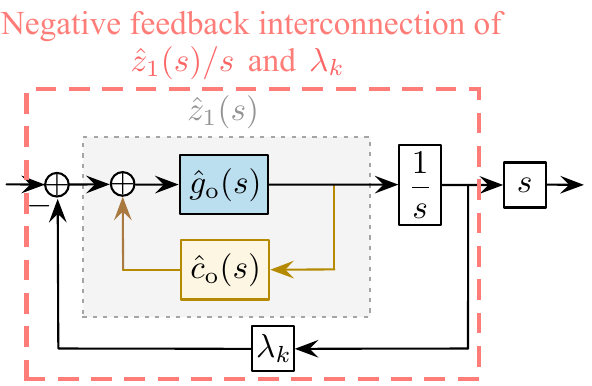}\label{fig:zk-v2}}
\caption{Equivalent block diagrams of $\hat{z}_{k}(s)$.}
\label{fig:combine simulation pdom np}
\end{figure}
Next, observe from \eqref{eqfeed-lambda} that $\mathcal{F}(\hat{z}_{1}(s)/s,\lambda_k)$ can be interpreted as a specific instance of the parametrized family $\mathcal{F}(\hat{z}_{1}(s)/s,\lambda)$ evaluated at $\lambda=\lambda_k$, which means that the poles of $\mathcal{F}(\hat{z}_{1}(s)/s,\lambda_k)$ (or $\hat{z}_{k}(s)$ by the preceding analysis) coincide exactly with the poles of $\mathcal{F}(\hat{z}_{1}(s)/s,\lambda)$ when $\lambda=\lambda_k$. Thus, two results follow.

First, since $0 = \lambda_1 < \lambda_2 \leq \ldots \leq\lambda_n$, the requirement on $\hat{c}_\mathrm{o}(s)$ mentioned above is ensured by our assumption that $\hat{c}_\mathrm{o}(s)$ is designed to make the root locus of $\mathcal{F}(\hat{z}_{1}(s)/s,\lambda)$ lie in the open left half-plane for any $\lambda>0$.

Second, if we understand how the poles of the closed-loop system $\mathcal{F}(\hat{z}_{1}(s)/s,\lambda)$ evolve as $\lambda$ increases from $\lambda_2$ to $\lambda_n$, we can directly determine where the poles of $\hat{z}_{k}(s)$ for $k\in\mathcal{N}\setminus\{1\}$ are located on the complex plane. Clearly, the root locus method offers a natural way to analyze this evolution graphically. Hence, the key step is to construct the root locus for $\mathcal{F}(\hat{z}_{1}(s)/s,\lambda)$ as $\lambda$ varies, which can be accomplished by considering the loop-gain $\hat{L}(s)$ in~\eqref{eq:Loopgain}.
\end{proof}

\subsection{Fine-Tuning Goal for Frequency Control}\label{ssec:tune-goal}
In the remainder of this manuscript, on top of the COI frequency Nadir elimination via FS established by our prior work~\cite{jiang2021tac,jiangtps2021}, we will derive a foolproof oscillation damping tuning rule by leveraging the root locus associated with \eqref{eq:Loopgain}. This allows us to simultaneously remove the COI frequency Nadir and damp out all other oscillatory modes to a specified extent following major disturbances, which proves challenging to accomplish through the commonly used VI. By doing so, the frequency security and oscillatory stability in power grids with high renewable integration can be greatly improved. 

To mathematically characterize this goal, we consider the scenario where the system in Fig.~\ref{fig:model} is subject to step power disturbances $\boldsymbol{p}=\boldsymbol{u}_0 \mathds{U}_{ t \geq 0 }$ with $\boldsymbol{u}_0 := \left(u_{0,i}, i \in \mathcal{N} \right) \in \real^n$ being an arbitrary vector direction that allows for power disturbances of different magnitudes at individual buses, i.e., $\hat{\boldsymbol{p}}(s)=\boldsymbol{u}_0/s$. Then, the frequency response in \eqref{eq:omega-t-decomple-coi} becomes
\begin{align*}
    \hat {\boldsymbol{\omega}}(s) =\dfrac{\mathbbold{1}_n^T\boldsymbol{u}_0}{\sum_{i=1}^n r_i}\dfrac{\hat{z}_{1}(s)}{s}\mathbbold{1}_n+\sum_{k=2}^n \dfrac{\hat{z}_{k}(s)}{s}\boldsymbol{R}^{-\frac{1}{2}} 
    \boldsymbol{v}_k\boldsymbol{v}_k^T\boldsymbol{R}^{-\frac{1}{2}}\boldsymbol{u}_0\,,
\end{align*}
whose time-domain counterpart is 
\begin{align}
    \boldsymbol{\omega}(t) =\bar {\omega}(t)\mathbbold{1}_n+\tilde {\boldsymbol{\omega}}(t)\label{eq:omega-t-decomple}
\end{align}
with
\begin{align}
    \bar {\omega}(t):=&\ \dfrac{\sum_{i=1}^n u_{0,i}}{\sum_{i=1}^n \!r_i}z_{\mathrm{u},1}(t)\,,\label{eq:w-bar-t}\\\tilde {\boldsymbol{\omega}}(t):=&\sum_{k=2}^n \!z_{\mathrm{u},k}(t)\boldsymbol{R}^{-\frac{1}{2}} 
    \boldsymbol{v}_k\boldsymbol{v}_k^T\boldsymbol{R}^{-\frac{1}{2}}\boldsymbol{u}_0\,,\label{eq:w-til-t}
\end{align}
where $z_{\mathrm{u},k}(t):= \mathscr{L}^{-1}\{\hat{z}_{k}(s)/s\}$ denotes the unit-step response of $\hat{z}_k(s)$. Hence, our tuning goal is twofold:
\begin{itemize}
    \item \textbf{Frequency Security:} Given the expected maximum magnitude of the net power imbalance $\Delta P$ (in $\SI{}{\pu}$), the COI frequency deviation stays within the acceptable range $\pm\Delta\omega_{ \mathrm{d}}$ (in $\SI{}{\pu}$), i.e., $|\bar{\omega}|_\infty := \max_{t\geq0} |\bar{\omega}(t)|\leq\Delta\omega_{ \mathrm{d}}$. For example, the maximum allowed quasi-steady-state frequency deviation for the European and Great Britain systems is \SI{\pm200}{\milli\hertz} \cite{UCTLbook,NG2016standard}.
    \item \textbf{Oscillatory Stability:} Any oscillation in $\tilde {\boldsymbol{\omega}}(t)$ must have at least a damping ratio of $\cos\psi_{\mathrm{d}}\in(0,1]$ and a decay rate $\alpha_{\mathrm{d}}>0$, that is, $\forall k\in\mathcal{N}\setminus\{1\}$, $\hat{z}_{k}(s)$ is $(\alpha_{\mathrm{d}}, \psi_{\mathrm{d}})$-stable. For instance, the Western Power requirement regarding oscillations introduced in Section~\ref{sec:intro} can be approximately quantified as the damping ratio $\cos\psi\geq 0.1$ and the halving time $1/\alpha\leq5$, corresponding to $(0.2, 84.3\degree)$-stability, where the halving time is approximated as the reciprocal of the decay rate $\alpha$ since $1/\alpha$ is the time needed for transient amplitude to drop to $e^{-\alpha(1/\alpha)}=e^{-1}\approx 36.8\%<50\%$ of its initial value.
\end{itemize}


\section{Frequency Shaping Control for Oscillation Damping}\label{sec:FS}

In this section, after briefly reviewing FS, we analyze oscillatory stability for the system under FS via the root locus method proposed in Section~\ref{ssec:locus-method}, which yields closed-form expressions of the limit damping ratio and decay rate of all oscillatory modes. This provides a quantitative characterization of how the network spectrum and parameter values influence oscillatory stability under FS, upon which we propose foolproof tuning rules for FS to simultaneously achieve specified frequency security and oscillatory stability defined in Section~\ref{ssec:tune-goal}. Additionally, we show that FS outperforms the traditional VI in terms of frequency convergence rate.\footnote{Note that depending on the specific inverter dynamics involved, we may add subscripts ``vi'' or ``fs'' in the name of a transfer function or variable without making a further declaration in the rest of this manuscript.}

\subsection{$(\alpha, \psi)$-Stability Analysis via Root Locus}

FS~\cite{jiang2021tac,jiang2021lcss,jiangtps2021} is our previously proposed strategy to meet the frequency security objective discussed in Section~\ref{ssec:tune-goal}, whose dynamics is as follows: 
\begin{dyn-i}[Frequency Shaping]\label{dyn-fs}
This control law shapes the COI frequency $\bar{\omega}_{\mathrm{fs}}(t)$ of the system shown in Fig.~\ref{fig:model} following step power disturbances into first-order response that naturally has no Nadir:
\begin{equation} \label{eq:dy-fs}
\hat{c}_{\mathrm{fs}}(s) := \frac{ d_\mathrm{t} }{\tau s+1}-\left(d_\mathrm{b} + d_\mathrm{t}\right) \,,
\end{equation}
where $d_\mathrm{b}>0$ is a tunable parameter that plays the role of inverter inverse droop.
\end{dyn-i}

However, the existing tuning recommendation for FS lacks deliberate consideration of oscillatory stability, which is an issue that is gaining more attention amid the expanding integration of renewable energy into the grid. Therefore, it is desirable to propose systematic tuning guidelines for FS to simultaneously achieve satisfactory frequency security and oscillatory stability. To this end, it is important to understand better how $(\alpha, \psi)$-stability is related to the system and control parameters, which can be analyzed through the root locus.

\begin{lem}[Root locus associated with $\hat{L}_{\mathrm{fs}}(s)$]\label{lem:rlocus-FS}
For the system in Fig.~\ref{fig:model} under Assumption~\ref{ass:proportion}, if inverters adopt FS, i.e., $\hat{c}_\mathrm{o}(s)=\hat{c}_{\mathrm{fs}}(s)$ in \eqref{eq:dy-fs}, then the root locus of the system with the open-loop transfer function 
\begin{align}\label{eq:Loopgain-fs}
    \hat{L}_{\mathrm{fs}}(s):= \dfrac{\lambda\hat{z}_{1,{\mathrm{fs}}}(s)}{s}=\dfrac{\lambda}{m s^2 + (d+d_\mathrm{b} + d_{\mathrm{t}})s}
\end{align}
is as shown in Fig.~\ref{fig:rlocus-fs}. Precisely, the locus starts from open-loop poles 
\begin{align}\label{eq:open-pole-fs}
    s_1 = 0\qquad\text{and}\qquad s_2 =-\dfrac{d+d_\mathrm{b}+d_\mathrm{t}}{m}\,,
\end{align} and leaves the real-axis at the point
\begin{align}\label{eq:sigma-fs}
    \sigma_\mathrm{a}=-\dfrac{d+d_\mathrm{b}+d_\mathrm{t}}{2m}\,.
\end{align}
\end{lem}
\begin{figure}[t!]
\centering
\includegraphics[width=0.5\columnwidth]{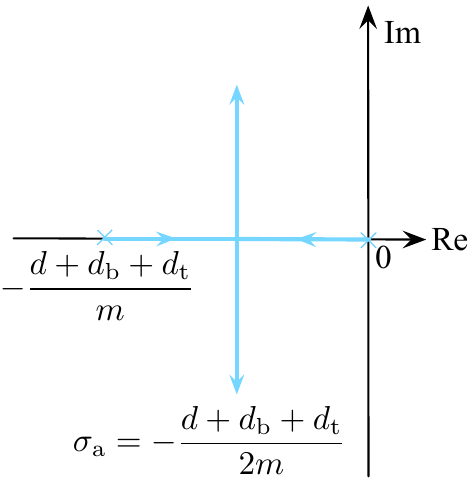}
\caption{Root locus of the system with loop-gain $\hat{L}_{\mathrm{fs}}(s)$ in \eqref{eq:Loopgain-fs}.}\label{fig:rlocus-fs}
\end{figure}
\begin{proof}
See the Appendix \ref{app:lem-rlocus-FS}.    
\end{proof}

The root locus in Fig.~\ref{fig:rlocus-fs} provided by Lemma~\ref{lem:rlocus-FS} makes $(\alpha, \psi)$-stability analysis of oscillatory modes under FS tractable, which combined with Theorem~\ref{thm:colocate-z} enables us to derive closed-form expressions of the minimum damping ratio and decay rate among all oscillatory subsystems summarized in the next theorem. Of course, all analysis relies on the premise required by Theorem~\ref{thm:colocate-z} that $\hat{c}_{\mathrm{fs}}(s)$ makes the root locus of $\mathcal{F}(\hat{z}_{1,{\mathrm{fs}}}(s)/s,\lambda)$ lie in the open left half-plane for any $\lambda>0$, which is clearly true in Fig.~\ref{fig:rlocus-fs}. 
\begin{thm}[Minimum damping ratio and decay rate under FS]\label{thm:osc-limit-fs}
For the system in Fig.~\ref{fig:model} under Assumption~\ref{ass:proportion}, if inverters adopt FS, i.e., $\hat{c}_\mathrm{o}(s)=\hat{c}_{\mathrm{fs}}(s)$ in \eqref{eq:dy-fs}, 
then the minimum damping ratio $\cos\overline{\psi}$ and the minimum decay rate $\underline{\alpha}$ among all $\hat{z}_{k,{\mathrm{fs}}}(s)$ for $k\in\mathcal{N}\setminus\{1\}$ can be explicitly determined as follows: 
\begin{itemize}
\item (minimum damping ratio)
\begin{align}\label{eq:min-damp-fs}
 \!\!\!\!\cos\overline{\psi}\!=\!\begin{cases}\dfrac{d\!+\!d_\mathrm{b}\!+\!d_\mathrm{t}}{2\sqrt{\lambda_n m}}&\!\!\textrm{if $0\leq d_\mathrm{b}< 2\sqrt{\lambda_n m}-d-d_\mathrm{t}$,}\\1&\!\!\textrm{otherwise;} 
 \end{cases}\!\!
\end{align}
\item (minimum decay rate)
\begin{align}
 &\underline{\alpha}\label{eq:min-decay-fs}\\&\!\!=\!\begin{cases}
\!\!\dfrac{d\!+\!d_\mathrm{b}\!+\!d_\mathrm{t}}{2m}\!=\!|\sigma_\mathrm{a}|\quad\!\!\textrm{if $0\leq d_\mathrm{b}\leq 2\sqrt{\lambda_2m}-d-d_\mathrm{t}$,}\nonumber\\
	\!\!\dfrac{d\!+\!d_\mathrm{b}\!+\!d_\mathrm{t}\!-\!\sqrt{(d\!+\!d_\mathrm{b}\!+\!d_\mathrm{t})^2\!-\!4\lambda_2m}}{2 m}\!<\!|\sigma_\mathrm{a}|\ \ \textrm{otherwise.}\!\!\!\!\!\nonumber
	\end{cases}   
\end{align}
\end{itemize}
\end{thm}
\begin{proof}
See the Appendix~\ref{app:osc-limit-fs-pf}.
\end{proof}

By quantifying the limit $(\underline{\alpha}, \overline{\psi})$-stability of all subsystems associated with oscillatory behavior, Theorem~\ref{thm:osc-limit-fs} offers valuable insight into how the network spectrum and parameter values of a power system influence oscillatory stability under FS. First and foremost, regarding the impact of the network spectrum, the damping ratio $\cos\overline{\psi}$ and decay rate $\underline{\alpha}$ are determined exclusively by the largest eigenvalue $\lambda_n$ and the Fiedler eigenvalue $\lambda_2$ of the scaled Laplacian matrix $\boldsymbol{L}$, respectively. Moreover, both $\cos\overline{\psi}$ and $\underline{\alpha}$ are piecewise functions of the inverter inverse droop $d_\mathrm{b}$, albeit with different partition points. To make these functional relations more accessible, we now visualize \eqref{eq:min-damp-fs} and \eqref{eq:min-decay-fs} as Fig.~\ref{fig:fs-metric-db} using parameters from the modified WSCC system. Some observations are in order. First, the damping ratio $\cos\overline{\psi}$ increases linearly with $d_\mathrm{b}$ at a rate of  $1/(2\sqrt{\lambda_n m})$ until $d_\mathrm{b}=2\sqrt{\lambda_n m}-d-d_\mathrm{t}$, which plays the role of a saturation point after which $\cos\overline{\psi}$ stays at $1$ even if $d_\mathrm{b}$ increases. This sensitivity related to $\lambda_n$ suggests that, for the system to be robust against frequency oscillations, the transmission network should be designed with small $\lambda_n$ which leads to a larger damping ratio. 
Second, the decay rate $\underline{\alpha}$ increases linearly with $d_\mathrm{b}$ at a rate of $1/(2m)$ until reaching its limit $\sqrt{\lambda_2 /m}$ when $d_\mathrm{b}=2\sqrt{\lambda_2 m}-d-d_\mathrm{t}$, beyond which it decreases monotonically.

\begin{figure}[t!]
\centering
\subfigure[Minimum damping ratio as a function of $d_\mathrm{b}$ ]
{\includegraphics[width=\columnwidth]{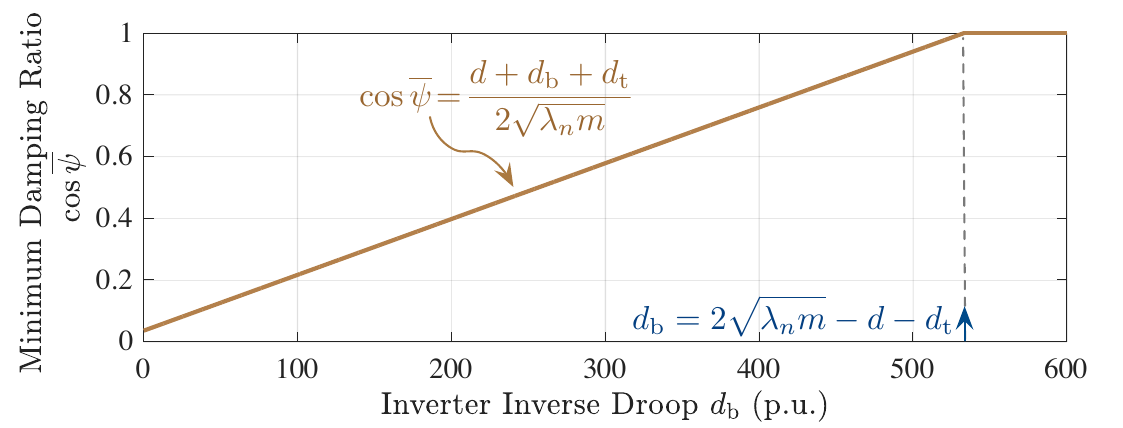}\label{fig:damping-db-fs}}
\subfigure[Minimum decay rate as a function of $d_\mathrm{b}$ ]
{\includegraphics[width=\columnwidth]{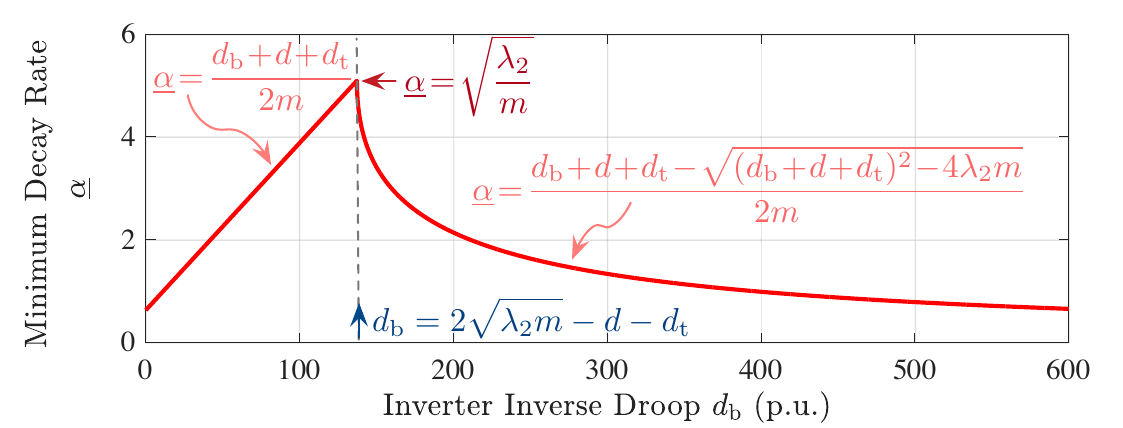}\label{fig:decay-db-fs}}
\hfil
\caption{Effect of the inverter inverse droop $d_\mathrm{b}$ under FS on the oscillation damping ratio and decay rate in the modified WSCC system.}
\label{fig:fs-metric-db}
\end{figure}

Thus, collectively, the two panels in Fig.~\ref{fig:fs-metric-db} indicate that the effect of $d_\mathrm{b}$ on oscillatory stability varies considerably across three regimes:
(i) If $d_\mathrm{b}\in[0,2\sqrt{\lambda_2 m}-d-d_\mathrm{t})$, increasing $d_\mathrm{b}$ leads to a simultaneous increase in $\cos\overline{\psi}$ and $\underline{\alpha}$; (ii) If $d_\mathrm{b}\in[2\sqrt{\lambda_2 m}-d-d_\mathrm{t},2\sqrt{\lambda_n m}-d-d_\mathrm{t})$, increasing $d_\mathrm{b}$ further raises $\cos\overline{\psi}$ but reduces $\underline{\alpha}$, indicating a trade-off between $\cos\overline{\psi}$ and $\underline{\alpha}$; (iii) If $d_\mathrm{b}\in[2\sqrt{\lambda_n m}-d-d_\mathrm{t},\infty)$, increasing $d_\mathrm{b}$ offers no benefit since $\cos\overline{\psi}$ stays at $1$ and $\underline{\alpha}$ decreases further. 

Finally, the limit $(\underline{\alpha}, \overline{\psi})$-stability of each $\hat{z}_{k,{\mathrm{fs}}}(s)$, $\forall k\in\mathcal{N}\setminus\{1\}$, established by Theorem~\ref{thm:osc-limit-fs}, combined with the intrinsic Nadir elimination property of FS, guarantees that the frequency deviations of the whole system converge exponentially fast to steady-state value with a rate of at least $\underline{\alpha}$. This
is captured by the following proposition.

\begin{prop}[Frequency convergence rate under FS]\label{pro:sys-rate}
Under Assumption~\ref{ass:proportion}, if the power system in Fig.~\ref{fig:model} with FS, i.e., $\hat{c}_\mathrm{o}(s)=\hat{c}_{\mathrm{fs}}(s)$ in \eqref{eq:dy-fs},  undergoes step power disturbances $\boldsymbol{p}=\boldsymbol{u}_0 \mathds{U}_{ t \geq 0 }$, then the frequency deviations $\boldsymbol{\omega}_{\mathrm{fs}}(t)=\bar{\omega}_{\mathrm{fs}}(t)\mathbbold{1}_n+\tilde {\boldsymbol{\omega}}_{\mathrm{fs}}(t)$ of the system synchronize to the steady-state value 
\begin{align}\label{eq:coi-fs-inf}
\bar{\omega}_{\mathrm{fs}}(\infty)=\dfrac{\sum_{i=1}^n u_{0,i}}{(d+d_\mathrm{b} + d_{\mathrm{t}})\sum_{i=1}^nr_i}    
\end{align} 
exponentially fast with a rate $\underline{\alpha}$ determined by \eqref{eq:min-decay-fs}, i.e.,
\begin{align}\label{eq:exp-wfs}
    \|\boldsymbol{\omega}_{\mathrm{fs}}(t)-\bar{\omega}_{\mathrm{fs}}(\infty)\mathbbold{1}_n\|<\delta e^{-\underline{\alpha}t}
\end{align}
for some constant $\delta>0$.
Particularly, 
\begin{itemize}
\item the COI frequency deviation $\bar{\omega}_{\mathrm{fs}}(t)$ evolves as 
\begin{align} \label{eq:coi-fs-t}
 \bar{\omega}_{\mathrm{fs}}(t)=\bar{\omega}_{\mathrm{fs}}(\infty)\left(1-e^{-\frac{d+d_\mathrm{b} + d_{\mathrm{t}}}{m}t}\right) \,,  
\end{align}
which converges to $\bar{\omega}_{\mathrm{fs}}(\infty)$ exponentially fast with a rate $(d+d_\mathrm{b}+d_\mathrm{t})/m>\underline{\alpha}$;
\item all frequency oscillations $\tilde {\boldsymbol{\omega}}_{\mathrm{fs}}(t)$ decay exponentially fast to zero with a rate of at least $\underline{\alpha}$. 
\end{itemize}
\end{prop}
\begin{proof}
See the Appendix~\ref{app:sys-rate-pf}.
\end{proof}

Proposition~\ref{pro:sys-rate} shows that the minimum decay rate $\underline{\alpha}$ among all subsystems bottlenecks the whole system convergence rate in a straightforward way: $\underline{\alpha}$ is exactly the system-wide convergence rate. Unfortunately, there is no single system-wide damping ratio in multi-input multi-output systems. That is why we settle for the second best by focusing on the mode with the minimum damping ratio $\cos\overline{\psi}$. This damping ratio $\cos\overline{\psi}$ characterized by Theorem~\ref{thm:osc-limit-fs} can be interpreted as a system robustness metric against oscillations. 
\begin{rem}[Discrepancy between pole-specific and system-level damping ratios in second-order systems] Even if each $\hat{z}_{k,{\mathrm{fs}}}(s)$ is a single-input single-output second-order system, there is subtle difference between the damping ratio of a pole and the damping ratio of the system in each $\hat{z}_{k,{\mathrm{fs}}}(s)$. The former quantifies the damping ratio of each pole individually via a geometric measure $\cos\psi$, while the latter is defined using characteristic polynomial coefficients. Their values match when a system has complex-conjugate poles or repeated real poles, corresponding to under damped case (damping ratio in $(0,1)$) and critically damped case (damping ratio $1$), respectively. However, for a system with distinct real poles, as long as a pole lies on the negative real-axis, its pole-specific damping ratio equals $\cos0\degree=1$, while the system-level damping ratio exceeds $1$, classifying the system as over damped. These two damping ratio concepts are not contradictory. Owing to its broader generalizability, the pole-specific damping ratio is adopted herein, as in practical applications such as MATLAB root locus plots.
\end{rem}
  
\subsection{Tuning Guidance for Desired Frequency Performance}\label{ssec:al-tune}
Recall from Section~\ref{ssec:tune-goal} that our goal is to tune FS for satisfactory frequency security and oscillatory stability, which is the core of this subsection. 

Suppose that the desired values of the damping ratio and decay rate are $\cos\psi_{\mathrm{d}}\in(0,1]$ and $\alpha_{\mathrm{d}}>0$, respectively. The information provided by Fig.~\ref{fig:fs-metric-db} enables us to formulate an easy-to-use oscillation damping tuning guide for FS to obtain satisfactory performance. Provided that $\cos\overline{\psi}$ and $\underline{\alpha}$ are functions of $d_\mathrm{b}$, in what follows we denote them by $\cos\overline{\psi}(d_\mathrm{b})$ and $\underline{\alpha}(d_\mathrm{b})$. Recall that the typical damping ratio $\cos\psi_{\mathrm{d}}$ required by grid codes is around $0.1$~\cite{WP2021}, which can be easily achieved by $d_\mathrm{b}\in[0,2\sqrt{\lambda_2 m}-d-d_\mathrm{t})$ according to Fig.~\ref{fig:fs-metric-db}. Thus, we henceforth restrict our discussion to the regime $d_\mathrm{b}\in[0,2\sqrt{\lambda_2 m}-d-d_\mathrm{t})$, in which both $\cos\overline{\psi}$ and $\underline{\alpha}$ increase linearly with $d_\mathrm{b}$. Admittedly, if the desired decay rate is so demanding that it exceeds the largest decay rate that FS can achieve, i.e., $\max_{d_\mathrm{b}}\underline{\alpha}(d_\mathrm{b})=\sqrt{\lambda_2/ m}\leq\alpha_{\mathrm{d}}$, then one has to resort to an alternative control method, which is out of the scope of this manuscript. Otherwise, the actual tuning of FS can be determined based on the existing system performance. Basically, if the existing system suffices to provide a satisfactory damping ratio and/or decay rate, then no additional compensation via $d_\mathrm{b}$ is needed for that particular oscillatory metric. That is, if
\begin{align*}
\cos\overline{\psi}(0)=\dfrac{d+d_\mathrm{t}}{2\sqrt{\lambda_n m}}\geq \cos\psi_{\mathrm{d}}\quad\text{or}\quad\underline{\alpha}(0)=\dfrac{d+d_\mathrm{t}}{2m}\geq\alpha_{\mathrm{d}} \,, 
\end{align*}
i.e.,
\begin{align}\label{eq:nece-comp-fs}
2\sqrt{\lambda_n m}\cos\psi_{\mathrm{d}}-d-d_\mathrm{t}\!\leq\! 0\ \ \text{or}\ \ 2m\alpha_{\mathrm{d}}-d-d_\mathrm{t}\!\leq\! 0\,,       \end{align}
then there is no need to introduce additional inverse droop $d_\mathrm{b}$ to improve $\cos\overline{\psi}$ or $\underline{\alpha}$, respectively. Once a deficiency in the damping ratio or decay rate is identified with the aid of \eqref{eq:nece-comp-fs}, the minimum required $d_\mathrm{b}$ to compensate for $\cos\overline{\psi}$ or $\underline{\alpha}$ can be solved from
\begin{align*}
    \cos\overline{\psi}(d_\mathrm{b})\!=\!\dfrac{d\!+\!d_\mathrm{b}\!+\!d_\mathrm{t}}{2\sqrt{\lambda_n m}}\!=\! \cos\psi_{\mathrm{d}}\ \text{or}\ \underline{\alpha}(d_\mathrm{b})\!=\!\dfrac{d\!+\!d_\mathrm{b}\!+\!d_\mathrm{t}}{2m}\!=\!\alpha_{\mathrm{d}}\,,
\end{align*}
respectively, which yields
\begin{align}\label{eq:comp-fs}
d_\mathrm{b}\!=\!2\sqrt{\lambda_n m}\cos\psi_{\mathrm{d}}-d-d_\mathrm{t}\ \text{or}\ d_\mathrm{b}\!=\!2m\alpha_{\mathrm{d}}-d-d_\mathrm{t}\,,       \end{align}
respectively. Combining above discussions related to \eqref{eq:nece-comp-fs} and \eqref{eq:comp-fs}, we can formulate the minimum required $d_\mathrm{b}$ to guarantee the desired damping ratio $\cos\psi_{\mathrm{d}}$ and decay rate $\alpha_{\mathrm{d}}$ compactly as:
\begin{align}\label{eq:db-osci}
    \!\!\!d_{\mathrm{b},\mathrm{osc}}\!:=\!\max \!\left(0,2\sqrt{\lambda_n m}\cos\psi_{\mathrm{d}}\!-\!d\!-\!d_\mathrm{t},2m\alpha_{\mathrm{d}}\!-\!d\!-\!d_\mathrm{t}\!\right)\!\!\,.\!\!
\end{align}

Notably, the tuning of $d_\mathrm{b}$ in \eqref{eq:db-osci} solely addresses the oscillatory stability. In general, it remains crucial to ensure frequency security as described by Section~\ref{ssec:tune-goal}. Similar to the tuning recommendation for FS in a single-machine power system~\cite{jiangtps2021}, this can be verified through simple algebraic calculations since FS enjoys the nice Nadir elimination property~\cite[Theorem 9]{jiang2021tac}. More precisely, if the power system in Fig.~\ref{fig:model} undergoes step power disturbances $\boldsymbol{p}=\boldsymbol{u}_0 \mathds{U}_{ t \geq 0 }$, then the minimum required $d_\mathrm{b}$ to ensure that $\bar{\omega}_{\mathrm{fs}}(t)$ stays within $\pm\Delta\omega_{ \mathrm{d}}$ when the magnitude of the net power imbalance $|\sum_{i=1}^n u_{0,i}|=\Delta P$ is given by 
\begin{align}\label{eq:db-ss}
    d_{\mathrm{b},\mathrm{COI}}:=\max \!\left(0,\frac{\Delta P}{(\sum_{i=1}^n r_i)\Delta\omega_{\mathrm{d}}}-d-d_\mathrm{t}\right)\,.
\end{align}
Here, the second term is compensation needed on the inverse droop to ensure frequency security, which can be easily determined using \eqref{eq:coi-fs-inf} and only takes effect when its value is positive.


To summarize, one can simply tune $d_{\mathrm{b}}$ for FS in \eqref{eq:dy-fs} as 
\begin{align}\label{eq:db-osci-coi}
    d_{\mathrm{b}}=\max \!\left(d_{\mathrm{b},\mathrm{COI}}, d_{\mathrm{b},\mathrm{osc}}\right)\,,
\end{align}
where $d_{\mathrm{b},\mathrm{COI}}$ determined from \eqref{eq:db-ss} ensures that the COI frequency stays within the acceptable range and $d_{\mathrm{b},\mathrm{osc}}$ determined from \eqref{eq:db-osci} ensures specified oscillatory stability. Note that, as mentioned at the beginning of this subsection, the underlying assumption for the tuning recommendation in \eqref{eq:db-osci-coi} is $d_{\mathrm{b},\mathrm{COI}}, d_{\mathrm{b},\mathrm{osc}}\in[0,2\sqrt{\lambda_2 m}-d-d_\mathrm{t})$. Nevertheless, this assumption can be moderately relaxed. For example, \eqref{eq:db-osci-coi} still works if $ d_{\mathrm{b},\mathrm{osc}}\in[0,2\sqrt{\lambda_2 m}-d-d_\mathrm{t})$ and $d_{\mathrm{b},\mathrm{COI}}\in[0,(m\alpha_{\mathrm{d}}^2+\lambda_2)/\alpha_{\mathrm{d}}-d-d_\mathrm{t})$. To see this, supposing that $d_{\mathrm{b},\mathrm{COI}}$ determined from \eqref{eq:db-ss} exceeds $(2\sqrt{\lambda_2 m}-d-d_\mathrm{t})$, we observe from Fig.~\ref{fig:fs-metric-db} that $d_{\mathrm{b}}$ computed from \eqref{eq:db-osci-coi} would lead to a damping ratio $\cos\overline{\psi}(d_{\mathrm{b},\mathrm{COI}})>\cos\psi_{\mathrm{d}}$ and a decay rate given by
\begin{align*} \dfrac{d+d_{\mathrm{b},\mathrm{COI}}+d_\mathrm{t}-\sqrt{(d+d_{\mathrm{b},\mathrm{COI}}+d_\mathrm{t})^2-4\lambda_2m}}{2 m}
\end{align*}
which is still satisfactory as long as its value is not less than $\alpha_{\mathrm{d}}$. Direct calculations show
that  
\begin{align*}
\dfrac{d+d_{\mathrm{b},\mathrm{COI}}+d_\mathrm{t}-\sqrt{(d+d_{\mathrm{b},\mathrm{COI}}+d_\mathrm{t})^2-4\lambda_2m}}{2 m}\geq \alpha_{\mathrm{d}}   
\end{align*}
is equivalent to 
\begin{align*}
   2\sqrt{\lambda_2 m}-d-d_\mathrm{t}\leq d_{\mathrm{b},\mathrm{COI}}\leq \dfrac{m\alpha_{\mathrm{d}}^2+\lambda_2}{\alpha_{\mathrm{d}}} -d-d_\mathrm{t}\,,
\end{align*}
which justifies the relaxation of $d_{\mathrm{b},\mathrm{COI}}$ to $d_{\mathrm{b},\mathrm{COI}}\in[0,(m\alpha_{\mathrm{d}}^2+\lambda_2)/\alpha_{\mathrm{d}}-d-d_\mathrm{t})$.

\subsection{Generalized Tuning via Visualization}
\label{ssec:visual}
The tuning \eqref{eq:db-osci-coi} recommended by the previous subsection focuses mainly on the linear region of Fig.~\ref{fig:fs-metric-db}, which should suffice for most practical needs. Nevertheless, for the sake of completeness, we would like to discuss more about how to take full advantage of Fig.~\ref{fig:fs-metric-db} by transforming it into Fig.~\ref{fig:achievable-fs}. 

\begin{figure}[t!]
\centering
\includegraphics[width=\columnwidth]{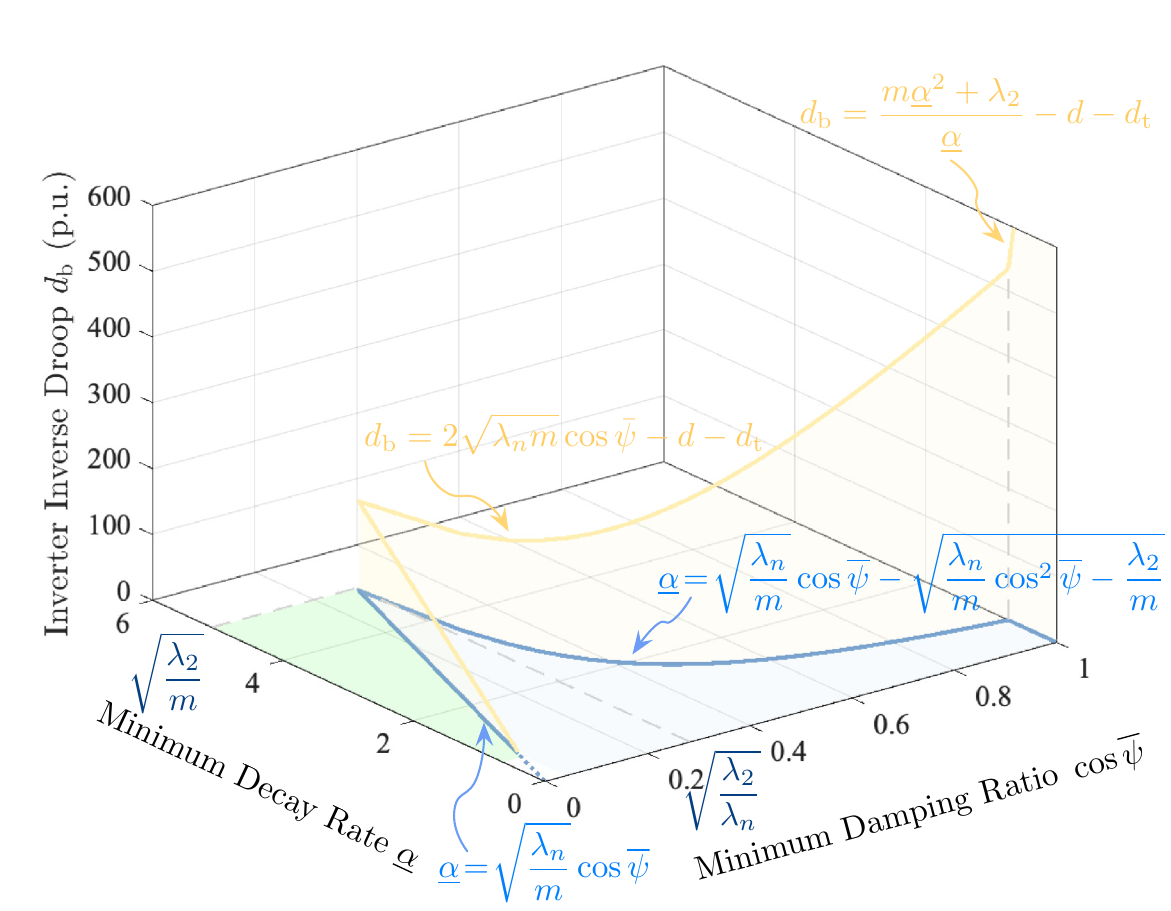}
\caption{Achievable combinations of the damping ratio $\cos\overline{\psi}$ and decay rate $\underline{\alpha}$ (solid blue line) as well as corresponding inverter inverse droop $d_{\mathrm{b}}$ (solid lemon line) under FS in the modified WSCC system.}\label{fig:achievable-fs}
\end{figure}

To obtain Fig.~\ref{fig:achievable-fs} from Fig.~\ref{fig:fs-metric-db}, the key is to derive the relation between $\underline{\alpha}$ and $\cos\overline{\psi}$ by eliminating the intermediate variable $d_{\mathrm{b}}$. This can be done by considering the cases $\cos\overline{\psi}\in [\cos\overline{\psi}(0),1)$ and $\cos\overline{\psi}=1$, separately.

If $\cos\overline{\psi}\in [\cos\overline{\psi}(0),1)$, then $d_{\mathrm{b}}$ is linear to $\cos\overline{\psi}$ via
\begin{equation}
d_\mathrm{b}=2\sqrt{\lambda_n m}\cos\overline{\psi}-d-d_\mathrm{t}\label{eq:db-damp-linear}
\end{equation}
according to Fig.~\ref{fig:damping-db-fs}. Substituting \eqref{eq:db-damp-linear} to the expression of $\underline{\alpha}$ in \eqref{eq:min-decay-fs} yields 
\begin{numcases}{\underline{\alpha}\!=\!}
   \!\!\sqrt{\dfrac{\lambda_n}{m}}\cos\overline{\psi}\quad\!\!\textrm{if $\dfrac{d+d_\mathrm{t}}{2\sqrt{\lambda_n m}}\leq \cos\overline{\psi}\leq \sqrt{\dfrac{\lambda_2}{\lambda_n}}$,} \rlap{\hspace{38pt}\setcounter{equation}{\value{equation}-1}\refstepcounter{equation}\label{eq:min-decay-damp-fs}(\theequation)}\nonumber 
   \\
   \!\!\sqrt{\dfrac{\lambda_n}{m}}\cos\overline{\psi}\!-\!\sqrt{ \dfrac{\lambda_n}{m}\cos^2\overline{\psi}\!-\!\dfrac{\lambda_2}{m}}\ \ \textrm{if $\sqrt{\dfrac{\lambda_2}{\lambda_n}}\!< \cos\overline{\psi}<1$.}\!\!\!\!\nonumber 
\end{numcases}

\addtocounter{equation}{1}

If $\cos\overline{\psi}=1$, theoretically, $\underline{\alpha}$ can be anything in 
\begin{equation}\label{eq:decay-damp1-fs}
\left.\left(0,\sqrt{\dfrac{\lambda_n}{m}}-\sqrt{\dfrac{\lambda_n-\lambda_2}{m}}\right.\right]    
\end{equation}
as long as we pick a finite $d_\mathrm{b}\geq2\sqrt{\lambda_n m}-d-d_\mathrm{t}$ properly. Specifically, the $d_\mathrm{b}$ leading to a particular $\underline{\alpha}$ can be determined from the nonlinear part of Fig.~\ref{fig:decay-db-fs} that falls into the second situation in \eqref{eq:min-decay-fs} as
\begin{align}\label{eq:db-damp-nonlinear}
    d_{\mathrm{b}}= \dfrac{m\underline{\alpha}^2+\lambda_2}{\underline{\alpha}} -d-d_\mathrm{t}\,.
\end{align}

Now, combining \eqref{eq:min-decay-damp-fs} and \eqref{eq:decay-damp1-fs} yields all achievable combinations of the damping ratio $\cos\overline{\psi}$ and decay rate $\underline{\alpha}$ under FS, plotted as the solid blue line in Fig.~\ref{fig:achievable-fs}. To achieve any possible combination along this solid blue line, the required inverter inverse droop $d_{\mathrm{b}}$ is visualized as the height of the point on the solid lemon line that is directly above that combination, which results from \eqref{eq:db-damp-linear} and \eqref{eq:db-damp-nonlinear}.  

Fig.~\ref{fig:achievable-fs} offers a visual tool to rapidly determine: (i) whether $\cos\overline{\psi}\geq\cos\psi_{\mathrm{d}}$ and $\underline{\alpha}\geq\alpha_{\mathrm{d}}$ can be guaranteed for any given $\cos\psi_{\mathrm{d}}$ and $\alpha_{\mathrm{d}}$; (ii) how much $d_{\mathrm{b}}$ is needed if the answer to the first question is yes. Notably, although achievable combinations of $\cos\overline{\psi}$ and $\underline{\alpha}$ are limited to the solid blue line, these combinations suffice to ensure $\cos\overline{\psi}\geq\cos\psi_{\mathrm{d}}$ and $\underline{\alpha}\geq\alpha_{\mathrm{d}}$ for a wide range of $\cos\psi_{\mathrm{d}}$ and $\alpha_{\mathrm{d}}$. Simply speaking, as long as the point $(\cos\psi_{\mathrm{d}},\alpha_{\mathrm{d}})$ representing the desired oscillatory stability lies within the blue or green shaded regions on the $(\cos\overline{\psi},\underline{\alpha})$-plane of Fig.~\ref{fig:achievable-fs}, we can guarantee $\cos\overline{\psi}\geq\cos\psi_{\mathrm{d}}$ and $\underline{\alpha}\geq\alpha_{\mathrm{d}}$ by projecting $(\cos\psi_{\mathrm{d}},\alpha_{\mathrm{d}})$ onto an appropriate point $(\cos\overline{\psi},\underline{\alpha})$ on the solid blue line such that $(\cos\overline{\psi},\underline{\alpha})\succeq(\cos\psi_{\mathrm{d}},\alpha_{\mathrm{d}})$, where $\succeq$ denotes component-wise ordering. For example, an easy way is to simply project any $(\cos\psi_{\mathrm{d}},\alpha_{\mathrm{d}})$ in the green shaded region along the positive $\cos\overline{\psi}$-axis and any $(\cos\psi_{\mathrm{d}},\alpha_{\mathrm{d}})$ in the blue shaded region along the positive $\underline{\alpha}$-axis to the solid blue line, which gives a satisfactory $(\cos\overline{\psi},\underline{\alpha})$ that can be achieved by setting $d_{\mathrm{b}}$ as the value suggested by the solid lemon line.

The value of $d_{\mathrm{b}}$ read from Fig.~\ref{fig:achievable-fs} through this procedure plays the role of $d_{\mathrm{b},\mathrm{osc}}$ in \eqref{eq:db-osci-coi}. Again, after determining $d_{\mathrm{b},\mathrm{osc}}$ for oscillatory stability, we still need to calculate $d_{\mathrm{b},\mathrm{COI}}$ using \eqref{eq:db-ss} for frequency security. Then $d_{\mathrm{b}}$ that satisfies both requirements is given by \eqref{eq:db-osci-coi}. Of course, it is important to check if the calculated $d_{\mathrm{b},\mathrm{COI}}$ would adversely affect oscillatory stability, primarily decay rate, ensured by $d_{\mathrm{b},\mathrm{osc}}$. It turns out that this visualized tuning process boasts broader applicability than the algebraic approach detailed in the previous subsection. Similar analysis as in the end of the previous subsection shows that \eqref{eq:db-osci-coi} is a valid design if $d_{\mathrm{b},\mathrm{osc}}$ and $d_{\mathrm{b},\mathrm{COI}}$ obtained here fall into one of three cases: (i) $d_{\mathrm{b},\mathrm{COI}}\leq d_{\mathrm{b},\mathrm{osc}}$; (ii) $d_{\mathrm{b},\mathrm{osc}}<d_{\mathrm{b},\mathrm{COI}}\leq 2\sqrt{\lambda_2 m}-d-d_\mathrm{t}$; (iii) $d_{\mathrm{b},\mathrm{osc}}\leq 2\sqrt{\lambda_2 m}-d-d_\mathrm{t}<d_{\mathrm{b},\mathrm{COI}}\leq(m\underline{\alpha}^2+\lambda_2)/\underline{\alpha}-d-d_\mathrm{t})$.

\subsection{Advantages over Virtual Inertia Control}

The most common inverter-based control law VI, which can also be tuned to remove the COI frequency Nadir to greatly improve frequency security based on our prior work: 

\begin{dyn-i}[Virtual Inertia]\label{dyn-vi}
This control law can provide additional  inertial response and droop capability:
\begin{equation} \label{eq:dy-vi}
\hat{c}_{\mathrm{vi}}(s) := -\left(m_{\mathrm{v}} s + d_{\mathrm{b}}\right)\,,
\end{equation}
where $m_{\mathrm{v}}\geq0$ is the virtual inertia constant (If $m_{\mathrm{v}}=0$, then \eqref{eq:dy-vi} reduces to pure droop control) and $d_{\mathrm{b}}>0$ is the inverter inverse droop. By~\cite[Theorem 1]{jiangtps2021}, when the representative inverter $\hat{c}_\mathrm{o}(s)=\hat{c}_{\mathrm{vi}}(s)$, the COI frequency $\bar{\omega}_{\mathrm{vi}}(t)$ of the system shown in Fig.~\ref{fig:model} following step power changes has no Nadir if parameters $(d_\mathrm{b}, m_{\mathrm{v}})$ satisfy
\begin{align}\label{eq:mvmin}
   m_{\mathrm{v}} \geq m_\mathrm{v,min} :=\tau\left(\!\sqrt{d_\mathrm{t}}+\sqrt{d+d_\mathrm{t}+d_{\mathrm{b}}}\right)^2\!-m\,.
\end{align}
\end{dyn-i}

Again, we would like to perform oscillatory stability analysis to VI with Nadir elimination tuning via the root locus method.
\begin{lem}[Root locus associated with $\hat{L}_{\mathrm{vi}}(s)$]\label{lem:rlocus-VI}
For the system in Fig.~\ref{fig:model} under Assumption~\ref{ass:proportion}, if inverters adopt VI with tuning that eliminates
the Nadir of COI frequency, i.e., $\hat{c}_\mathrm{o}(s)=\hat{c}_{\mathrm{vi}}(s)$ in \eqref{eq:dy-vi} with $(d_{\mathrm{b}}, m_{\mathrm{v}})$ satisfying \eqref{eq:mvmin}, then the root locus of the system with the open-loop transfer function 
\begin{align}\label{eq:Loopgain-vi}
    \!\!\!\!\hat{L}_{\mathrm{vi}}(s):= \dfrac{\lambda\hat{z}_{1,{\mathrm{vi}}}(s)}{s}
    =
    \frac{\lambda}{(m+m_{\mathrm{v}})}\dfrac{s + \tau^{-1}}{s\left( s^2 + 2\xi\omega_\mathrm{n} s + \omega_\mathrm{n}^2 \right)}
\end{align}
with 
\begin{subequations}\label{eq:ex-xi-wn}
\begin{align}
\omega_\mathrm{n} :=& \ \sqrt{\cfrac{d+d_{\mathrm{b}} +d_{\mathrm{t}}}{(m+m_{\mathrm{v}})\tau}}\label{eq:wn}\\\xi :=&\  \dfrac{\tau^{-1}+\left(d+d_{\mathrm{b}}\right)/(m+m_{\mathrm{v}})}{2\sqrt{\left(d+d_{\mathrm{b}} +d_{\mathrm{t}}\right)/\left[(m+m_{\mathrm{v}})\tau\right]}} 
\end{align}
\end{subequations}
is as shown in Fig.~\ref{fig:rlocus-vi}.
\end{lem}

\begin{figure}[t!]
\centering
\subfigure[$\xi=1$]
{\includegraphics[width=0.49\columnwidth]{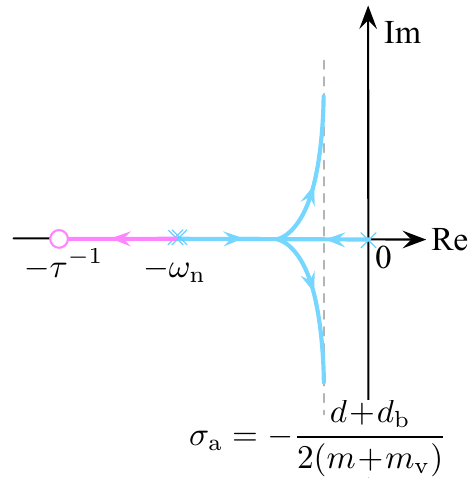}\label{fig:rlocus-vi-cri}}
\subfigure[$\xi>1$]
{\includegraphics[width=0.49\columnwidth]{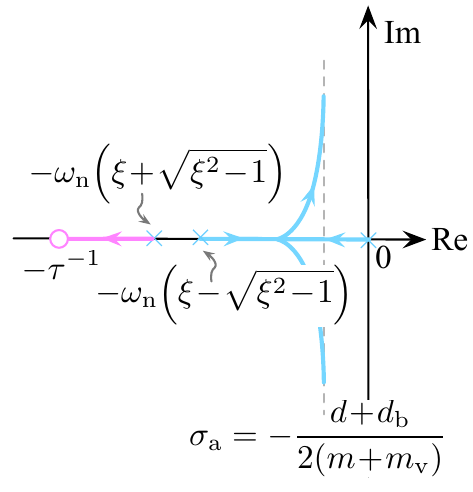}\label{fig:rlocus-vi-over}}
\caption{Root locus of the system with loop-gain $\hat{L}_{\mathrm{vi}}(s)$ in \eqref{eq:Loopgain-vi}.}\label{fig:rlocus-vi}
\end{figure}
\begin{proof}
See the Appendix~\ref{app:lem-rlocus-VI}. 
\end{proof}

Fig.~\ref{fig:rlocus-vi} provided by Lemma~\ref{lem:rlocus-VI} confirms that $\hat{c}_{\mathrm{vi}}(s)$ with Nadir elimination tuning \eqref{eq:mvmin} ensures that the root locus of $\mathcal{F}(\hat{z}_{1,{\mathrm{vi}}}(s)/s,\lambda)$ lies in the open left half-plane for any $\lambda>0$. Thus, $\forall k\in\mathcal{N}\setminus\{1\}$, poles of $\hat{z}_{k,{\mathrm{vi}}}(s)$ reside on the root locus in Fig.~\ref{fig:rlocus-vi} by Theorem~\ref{thm:colocate-z}, which, in principle, would enable the same $(\alpha, \psi)$-stability analysis technique used under FS. However, the increased complexity of root locus makes it intractable to derive closed-form expressions for $\cos\overline{\psi}$ and $\underline{\alpha}$ under VI. Nonetheless, qualitative insight into oscillatory stability can still be gleaned from the root locus shape. For example, the damping ratio $\cos\overline{\psi}$ remains exclusively determined by the largest eigenvalue $\lambda_n$ of the scaled Laplacian matrix $\boldsymbol{L}$, whereas the decay rate $\underline{\alpha}$ now depends on whichever of $\lambda_2$ or $\lambda_n$ lies closer to the imaginary-axis. Furthermore, although we cannot provide a precise characterization of frequency convergence rate as for FS, deriving a loose bound through Fig.~\ref{fig:rlocus-vi} for this rate under VI is straightforward. 

\begin{prop}[Frequency convergence rate under VI]\label{pro:sys-rate-vi}
Under Assumption~\ref{ass:proportion}, if the power system in Fig.~\ref{fig:model} with VI under Nadir elimination tuning, i.e., $\hat{c}_\mathrm{o}(s)=\hat{c}_{\mathrm{vi}}(s)$ in \eqref{eq:dy-vi} with $(d_{\mathrm{b}}, m_{\mathrm{v}})$ satisfying \eqref{eq:mvmin},  undergoes step power disturbances $\boldsymbol{p}=\boldsymbol{u}_0 \mathds{U}_{ t \geq 0 }$, then the frequency deviations $\boldsymbol{\omega}_{\mathrm{vi}}(t)=\bar{\omega}_{\mathrm{vi}}(t)\mathbbold{1}_n+\tilde {\boldsymbol{\omega}}_{\mathrm{vi}}(t)$ of the system synchronize to the steady-state value 
\begin{align}\label{eq:coi-vi-inf}
\bar{\omega}_{\mathrm{vi}}(\infty)=\dfrac{\sum_{i=1}^n u_{0,i}}{(d+d_\mathrm{b} + d_{\mathrm{t}})\sum_{i=1}^nr_i}    
\end{align} 
exponentially fast with a rate not exceeding $\omega_\mathrm{n}$ in \eqref{eq:wn}, that is,
\begin{align}\label{eq:exp-wvi}
    \|\boldsymbol{\omega}_{\mathrm{vi}}(t)-\bar{\omega}_{\mathrm{vi}}(\infty)\mathbbold{1}_n\|<\sigma e^{-\rho t}
\end{align}
for some constant $\sigma>0$ and $\rho\in(0,\omega_\mathrm{n}]$.
\end{prop}
\begin{proof}
See the Appendix~\ref{app:sys-rate-vi-pf}.
\end{proof}

Proposition~\ref{pro:sys-rate-vi} shows that, under VI, frequency deviations converge to the steady-state value at a rate not exceeding $\omega_\mathrm{n}$, which, in practice, is much more sluggish than under FS. To see this, recall that, in many practical scenarios, the frequency convergence rate under FS is linear in inverter inverse droop $d_\mathrm{b}$, i.e., $\underline{\alpha}=(d+d_\mathrm{b}+d_\mathrm{t})/(2m)$, which turns out to be greater than $\omega_\mathrm{n}$ with typical parameter values being considered. We now illustrate this in detail. Simple algebraic calculations show that $(d+d_\mathrm{b}+d_\mathrm{t})/(2m)>\omega_\mathrm{n}$ if and only if 
\begin{align}\label{eq:rate-compare}
\dfrac{\sqrt{d+d_\mathrm{b}+d_\mathrm{t}}\sqrt{(m+m_{\mathrm{v}})\tau}}{m}>2\,.    
\end{align}
By \eqref{eq:mvmin}, $m+m_{\mathrm{v}}\geq\tau\left(\!\sqrt{d_\mathrm{t}}+\sqrt{d+d_\mathrm{t}+d_{\mathrm{b}}}\right)^2$. Thus, the left-hand side of \eqref{eq:rate-compare} satisfies
\begin{align}\label{eq:rate-compare-nadir}
&\dfrac{\sqrt{d+d_\mathrm{b}+d_\mathrm{t}}\sqrt{(m+m_{\mathrm{v}})\tau}}{m}\\&\geq\!\dfrac{\sqrt{d+d_\mathrm{b}+d_\mathrm{t}}\tau\left(\!\sqrt{d_\mathrm{t}}+\sqrt{d+d_\mathrm{t}+d_{\mathrm{b}}}\right)}{m}
\nonumber\\&=\!\dfrac{\tau\left(\sqrt{d+d_\mathrm{b}+d_\mathrm{t}}\sqrt{d_\mathrm{t}}+d+d_\mathrm{t}+d_{\mathrm{b}}\right)}{m}\!>\!\dfrac{\tau\left(d+2d_\mathrm{t}+d_{\mathrm{b}}\right)}{m}\nonumber\,.    
\end{align}
Typically, generator inertia constant $m$ is within $\SI{20}{\second}$~\cite[Table 3.2]{kundur_power_1994}, turbine time constants $\tau$ range from sub-second
to several seconds~\cite[Chapter 9]{kundur_power_1994}, and turbine inverse droop $d_\mathrm{t}$ is in the interval of $10$--$\SI{20}{\pu}$~\cite{Vorobev2019tps}. Now, if we consider the extreme case where $m=\SI{20}{\second}$ and $d_\mathrm{t}=\SI{10}{\pu}$, we know 
\begin{align} \label{eq:bound-tau-d}\dfrac{\tau\left(d+2d_\mathrm{t}+d_{\mathrm{b}}\right)}{m}\geq\tau\left(1+\dfrac{d+d_{\mathrm{b}}}{20}\right)\,,
\end{align}
where the right-hand side is greater than $2$ for any $\tau\geq \SI{2}{\second}$ or $d_{\mathrm{b}}\geq \SI{20}{\pu}$. In realistic systems where $m$ is often lower or $d_\mathrm{t}$ is often larger, the left-hand side of \eqref{eq:bound-tau-d} can readily exceed $2$ even when the turbine dynamics is faster or inverse droop is smaller. For this reason, combined with \eqref{eq:rate-compare-nadir}, we conclude that \eqref{eq:rate-compare} holds practically, which indicates that the frequency converges more rapidly under FS than under VI.

Therefore, FS is a better choice for frequency control compared to VI mainly for two reasons. First, FS allows for foolproof tuning to simultaneously shape COI frequency nicely and damp oscillations to a specified extent, which is hard to realize by VI due to the absence of closed-form expressions for $(\alpha, \psi)$-stability analysis. Second, frequency generally converges much faster under FS than under VI when COI frequency is shaped into Nadir-less response with the same steady-state value under both controllers.   

\section{Numerical Illustrations}\label{sec:simulation}
In this section, we provide numerical illustrations for the tuning process and system performance under FS. 

The simulations are performed on the WSCC $9$-bus $3$-generator system in Fig.~\ref{fig:WCSS_test_case} taken from Power System Toolbox (PST)~\cite{chow1992toolbox}. In order to mimic the weakly connected network, we modify the existing WSCC test case by increasing the impedance of the lines $4-9$ and $5-6$ to $20$ times its value, which results in a distinct $2$-area system. The dynamic model is then built upon the
Kron reduced system where only the 3 generator buses
are retained. Even though our previous analysis is based on the proportionality assumption (Assumption \ref{ass:proportion}), the simulations are still conducted with the heterogeneous generator inertia and damping coefficients directly obtained from the dataset, i.e., $m_1 = \SI{27.28}{\second}$, $m_2 = \SI{12.80}{\second}$, $m_3 = \SI{6.02}{\second}$, $d_1 = \SI{9.6}{\pu}$, $d_2 = \SI{2.5}{\pu}$, $d_3 = \SI{1}{\pu}$. Given that the turbine coefficients are not provided by the dataset, we set the turbine inverse droops to be equal, i.e., $d_{\mathrm{t},1}=d_{\mathrm{t},2}=d_{\mathrm{t},3}=\SI{15}{\pu}$, but choose the turbine time constants to be somewhat heterogeneous to make the test case more
realistic, i.e., $\tau_1 = \SI{2.80}{\second}$, $\tau_2 = \SI{2.10}{\second}$, $\tau_3 = \SI{1.66}{\second}$. Clearly, the proportionality assumption required in theoretical analysis is violated here. 

As discussed in Section~\ref{ssec:tune-goal}, our design goal is to ensure that, when this system experiences the maximum expected net power imbalance $\Delta P=\SI{0.2}{\pu}$, the COI frequency deviation stays within the acceptable range $\SI{\pm200}{\milli\hertz}$ and all oscillatory modes have $(0.2, 84.3\degree)$-stability, i.e. $\cos\psi_{\mathrm{d}}=0.1$ and $\alpha_{\mathrm{d}}=0.2$. All per unit values are on the system base with power base $S_\mathrm{B}=\SI{100}{\mega\VA}$ and nominal frequency $F_0=\SI{60}{\hertz}$. Clearly, this goal cannot be achieved by the original system according to Fig.~\ref{fig:fre-SG-weak}. Therefore, we add an inverter to each bus $i$, whose control law is either VI or FS introduced in Section~\ref{sec:FS}, both with Nadir elimination capability. The inverter tuning relies on the knowledge
of a representative generator in some form although the proportionality
assumption is violated. Here, we define the representative generator inertia constant
as the mean of individual generator inertia constants, i.e., $m := (\sum_{i=1}^3 m_i)/3=\SI{15.37}{\second}$. Accordingly, the proportionality parameters are given by $r_i := m_i/m$. Then, we define the representative damping coefficient as $d := (\sum_{i=1}^3 d_i)/(\sum_{i=1}^3 r_i)= \SI{4.37}{\pu}$ and the representative turbine inverse droop and time constant as $d_{\mathrm{t}} := (\sum_{i=1}^3 d_{\mathrm{t},i})/(\sum_{i=1}^3 r_i)=\SI{15}{\pu}$ and $\tau := (\sum_{i=1}^{3} \tau_i)/3=\SI{2.19}{\second}$, respectively. 

Regarding the tuning of FS, since the COI frequency of the system with $d_\mathrm{b}=\SI{0}{\pu}$ shown in Fig.~\ref{fig:fre-VIFS-notune-weak} is already satisfactory, we set $d_{\mathrm{b},\mathrm{COI}}=\SI{0}{\pu}$ and focus on determining $d_{\mathrm{b},\mathrm{osc}}$ to ensure $\cos\psi_{\mathrm{d}}=0.1$ and $\alpha_{\mathrm{d}}=0.2$. This can be easily done using the algebraic method proposed in Section~\ref{ssec:al-tune} or visualized method proposed in Section~\ref{ssec:visual}. Here, we adopt the algebraic one, where $d_{\mathrm{b},\mathrm{osc}}$ can be determined from \eqref{eq:db-osci} as $d_{\mathrm{b},\mathrm{osc}}=\max\left(0,35.89,-13.22\right)=\SI{35.89}{\pu}$. Thus, $d_{\mathrm{b}}=\max \left(0, 35.89\right)=\SI{35.89}{\pu}$ by \eqref{eq:db-osci-coi}. The performance of the system in this case when a $\SI{-0.2}{\pu}$ step power change is introduced to bus $1$ at $t=\SI{1}{\second}$ is shown in Fig.~\ref{fig:damp3589}, where frequency deviations including oscillations settle down extremely fast. To show the advantage of FS, the performance of the system under VI following the same contingency is provided in Fig.~\ref{fig:VI3589} as well, where we still set $d_{\mathrm{b}}=\SI{35.89}{\pu}$ to ensure that steady-state frequency
deviations under VI and FS are the same for a fair comparison and set $m_{\mathrm{v}} = m_\mathrm{v,min}=\SI{264.16}{\second}$ to eliminate Nadir with the smallest control effort. Clearly, FS outperforms VI since the frequency deviations converge exponentially fast with
a much higher rate and lower inverter power output under FS than
VI.

\begin{figure}[t!]
\centering
\subfigure[System under VI with $d_{\mathrm{b}}=\SI{35.89}{\pu}$ and $ m_{\mathrm{v}} = m_\mathrm{v,min}=\SI{264.16}{\second}$]
{\includegraphics[width=\columnwidth]{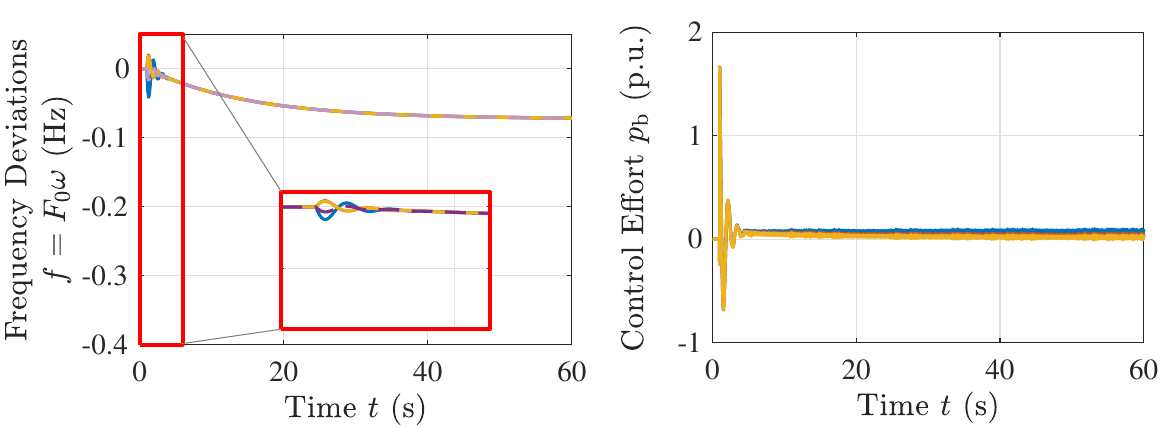}\label{fig:VI3589}}
\hfil
\subfigure[System under FS with $d_{\mathrm{b}}=\SI{35.89}{\pu}$]
{\includegraphics[width=\columnwidth]{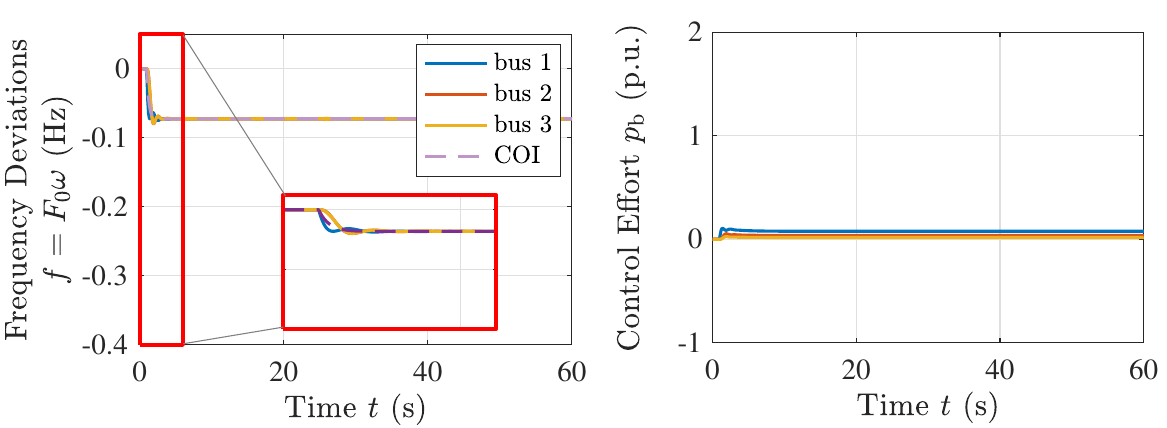}\label{fig:damp3589}}
\caption{Performance of the modified WSCC system when a $\SI{-0.2}{\pu}$ step power change is introduced to bus $1$ at $t=\SI{1}{\second}$.}
\label{fig:compareS}
\end{figure}



\section{Conclusions and Outlook}\label{sec:conclusion}

To better understand and address inter-area oscillations in weak grids, we have performed a systematic analysis of damping ratio and decay rate of inter-area
oscillations under FS via an efficient root locus approach. This offers many insights, one of which is that the minimum damping ratio is
determined by the largest eigenvalue of the scaled network Laplacian. Those analytic results enable us to propose foolproof
fine-tuning guidelines for FS to simultaneously meet specified metrics for frequency
security and oscillatory stability. Precisely, with our proposed tuning, FS can not only shape the COI frequency into a first-order response that converges to a steady-state value within the allowed range but also ensure a satisfactory damping ratio and decay rate of inter-area oscillations. Moreover, FS outperforms VI significantly, requiring less inverter power output to ensure a faster frequency convergence rate. Our future research will concentrate on some important extensions,
including analysis of impact of grid-forming inverters, lossy transmission lines, and heterogeneous parameters.



\appendix[]
\subsection{Proof of Lemma~\ref{lem:rlocus-FS}} \label{app:lem-rlocus-FS}
Applying \eqref{eq:go} and $\hat{c}_\mathrm{o}(s)=\hat{c}_{\mathrm{fs}}(s)$ in \eqref{eq:dy-fs} to \eqref{eq:zk-s} for $k=1$ shows that
\begin{align}\label{eq:z1fs}
    \hat{z}_{1,{\mathrm{fs}}}(s)=\dfrac{1}{m s + d+d_\mathrm{b} + d_{\mathrm{t}}}\,,
\end{align}
which directly gives $\hat{L}_{\mathrm{fs}}(s)$ in \eqref{eq:Loopgain-fs}. Now, the root locus of the system with the loop-gain $\hat{L}_{\mathrm{fs}}(s)$ can be plotted using the standard rules~\cite[Module 10]{Driels1996linear}.

First, from $m s^2 + (d+d_\mathrm{b} + d_{\mathrm{t}})s=ms[s+(d+d_\mathrm{b}+d_\mathrm{t})/m]=0$, we know that there are two open-loop poles given by \eqref{eq:open-pole-fs}, which are denoted by a cross in Fig.~\ref{fig:rlocus-fs}. Clearly, there is no open-loop zero. Basically, the locus starts at open-loop poles for $\lambda=0$ and both branches of the locus go to infinity when $\lambda=\infty$ since $N_\mathrm{p}-N_\mathrm{z}=2-0=2$, where $N_\mathrm{p}$ and $N_\mathrm{z}$ are the number of open-loop poles and zeros, respectively.

The real-axis segment forming part of the locus is to the left of an odd number of poles or zeros, which is only between the two poles $s_1$ and $s_2$ in this case.

To determine how the poles go to infinity as $\lambda$ increases, we next compute the center and orientation of the asymptotes. Precisely, the asymptotes intersect the real-axis at
\begin{align*}
    \sigma_\mathrm{a} 
    =&\ \frac{\text{Re}(s_1)+\text{Re}(s_2)}{N_\mathrm{p}-N_\mathrm{z}}\\
    =&\  \frac{0+[-(d+d_\mathrm{b}+d_\mathrm{t})/m]}{2}
    = -\dfrac{d+d_\mathrm{b}+d_\mathrm{t}}{2m}
\end{align*}
and incline to the positive real-axis at
\begin{align*}
\dfrac{180\degree(1+2l)}{N_\mathrm{p}-N_\mathrm{z}}\quad\text{with }l=0,1,\ldots, (N_\mathrm{p}-N_\mathrm{z}-1) \,, 
\end{align*}
which yields $90\degree$ and $270\degree$ here for $N_\mathrm{p}-N_\mathrm{z}=2$. 

At this stage, it seems likely that, as $\lambda$ increases, the closed-loop poles move from the open-loop poles toward $\sigma_\mathrm{a}$ along the real-axis and then branch off to infinity along asymptotes that are parallel to the imaginary-axis. This can be checked by determining the break-away point from the real-axis. To find such a break point, we need to look for the local maximizer of the gain $\lambda$ as a function of $s$. With this aim, we write the characteristic equation
\begin{align}\label{eq:char-FS}
 0=1+\hat{L}_{\mathrm{fs}}(s)=1+\dfrac{\lambda}{m s^2 + (d+d_\mathrm{b} + d_{\mathrm{t}})s}   
\end{align}
obtained from \eqref{eq:Loopgain-fs} in the form where $\lambda$ is expressed directly as a function of $s$:
\begin{align}\label{eq:lambda-s}
    \lambda=-[m s^2 + (d+d_\mathrm{b} + d_{\mathrm{t}})s]\,.
\end{align}
Then, \eqref{eq:lambda-s} can be differentiated and equated to zero as
\begin{align*}
    \dfrac{\mathrm{d}\lambda}{\mathrm{d}s}=-2ms-(d+d_\mathrm{b} + d_{\mathrm{t}})=0\,,
\end{align*}
from which we can solve for $s$ producing the maximum value of $\lambda$ as 
\begin{align*}
    s=-\dfrac{d+d_\mathrm{b}+d_\mathrm{t}}{2m}=\sigma_\mathrm{a}\,.
\end{align*}
Hence, the only break point on the real-axis segment is exactly the center of the asymptotes, which confirms that the final locus is as shown in Fig.~\ref{fig:rlocus-fs}.

\subsection{Proof of Theorem~\ref{thm:osc-limit-fs}}\label{app:osc-limit-fs-pf}
The whole proof leverages the fact that the gain $\lambda$ at a selected point $s$ on the locus in Fig.~\ref{fig:rlocus-fs} can be evaluated by measuring the length of each line joining the point $s$ to all open-loop poles $s_1$ and $s_2$ in \eqref{eq:open-pole-fs}~\cite[Module 10]{Driels1996linear}. To see this, observe from the characteristic equation \eqref{eq:char-FS} in the proof of Lemma~\ref{lem:rlocus-FS} that a particular point $s$ on the locus must satisfy the following magnitude equation
\begin{align*}
 \dfrac{\lambda}{m|s||s+(d+d_\mathrm{b}+d_\mathrm{t})/m|} =1  \,,
\end{align*}
which can be written as
\begin{align}\label{eq:K-FS}
 \lambda=m|s|\left|  s+\dfrac{d+d_\mathrm{b}+d_\mathrm{t}}{m}\right|  \,.
\end{align}
Thus, the value of gain $\lambda$ at any point $s$ on the locus in Fig.~\ref{fig:rlocus-fs} is just the product of distances between $s$ with the open-loop poles $s_1$ and $s_2$ scaled by $m$.

To streamline the proof, we start with the more interesting scenario where the maximum eigenvalue $\lambda_n$ of the scaled Laplacian is large enough such that
\begin{align}\label{eq:osi-eigmax}
    \lambda_n>m\sigma_\mathrm{a}^2=\dfrac{\left(d+d_\mathrm{b}+d_\mathrm{t}\right)^2}{4m}\,,
\end{align}
or equivalently 
\begin{align}\label{eq:cri-d-damping}
 d_\mathrm{b}< 2\sqrt{\lambda_n m}-d-d_\mathrm{t}\,.  
\end{align}
The other scenario where $\lambda_n\leq m\sigma_\mathrm{a}^2$, i.e., $d_\mathrm{b}\geq 2\sqrt{\lambda_n m}-d-d_\mathrm{t}$, is trivial and thus will be discussed at the end of the whole proof.

With the aid of \eqref{eq:K-FS}, we first illustrate the physical meaning of the assumption on $\lambda_n$ in \eqref{eq:osi-eigmax}. Note that the right-hand side of \eqref{eq:osi-eigmax} is just the value of gain $\lambda$ at the center of asymptotes, i.e., $s=\sigma_\mathrm{a}$ in \eqref{eq:sigma-fs}, since 
\begin{align*}
m\sigma_\mathrm{a}^2&=m|\sigma_\mathrm{a}||\sigma_\mathrm{a}-2\sigma_\mathrm{a}|=m|\sigma_\mathrm{a}|\left|\sigma_\mathrm{a}-2\left(-\dfrac{d+d_\mathrm{b}+d_\mathrm{t}}{2m}\right)\right| \\&=m|\sigma_\mathrm{a}|\left|\sigma_\mathrm{a}+\dfrac{d+d_\mathrm{b}+d_\mathrm{t}}{m}\right|   
\end{align*}
is exactly \eqref{eq:K-FS} evaluated at $s=\sigma_\mathrm{a}$. Therefore, \eqref{eq:osi-eigmax} simply means that the gain $\lambda_n$ occurs only after the closed-loop poles move beyond $s=\sigma_\mathrm{a}$ and thus away from the real-axis as illustrated in Fig.~\ref{fig:rlocus-fs-maxeig}. Actually, this assumption is minor since otherwise all gains from $\lambda_2$ to $\lambda_n$ would occur along the real-axis segment as shown in Fig.~\ref{fig:rlocus-fs-maxeig-violate}, which further would imply that all poles of $\hat{z}_{k,{\mathrm{fs}}}(s)$, for $k\in\mathcal{N}\setminus\{1\}$, would lie on the real-axis and thus would not induce any oscillatory behavior at all. In short, \eqref{eq:osi-eigmax} restricts our attention to power networks that display oscillatory behaviors. 

\begin{figure}[t!]
\centering
\subfigure[$\lambda_n$ not on real-axis if $\lambda_n>m\sigma_\mathrm{a}^2$]
{\includegraphics[width=0.49\columnwidth]{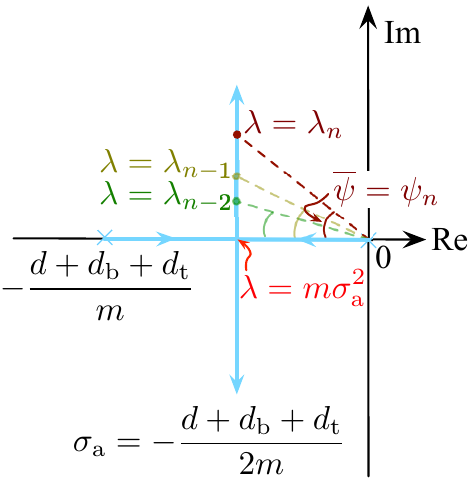}\label{fig:rlocus-fs-maxeig}}
\subfigure[$\lambda_n$ on real-axis if $\lambda_n\leq m\sigma_\mathrm{a}^2$]
{\includegraphics[width=0.49\columnwidth]{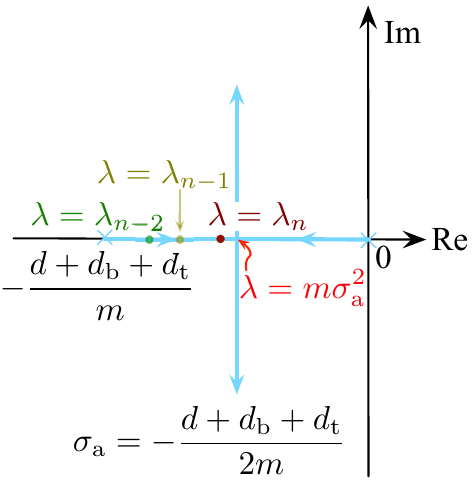}\label{fig:rlocus-fs-maxeig-violate}}
\caption{Interpretation of the minor assumption \eqref{eq:osi-eigmax} or \eqref{eq:cri-d-damping}.\protect\footnotemark}
\end{figure}
\footnotetext{This particular root locus enjoys symmetry with respect to both the real-axis and asymptote, implying that there are always a pair of closed-loop poles correspond to any gain $\lambda$. However, only one of two poles is typically marked for clean presentation unless necessary.}

Now, we are ready to quantify the minimum damping ratio $\cos\overline{\psi}$ and the minimum convergence rate $\underline{\alpha}$.

We begin with $\cos\overline{\psi}$. As discussed above, the assumption \eqref{eq:osi-eigmax} indicates that the gain $\lambda_n$ occurs on the asymptotes rather than the real-axis. In this case, it is easy to observe from Fig.~\ref{fig:rlocus-fs-maxeig} that, for any $k\in\mathcal{N}\setminus\{1\}$, the angle $\psi_k$ between the line joining the closed-loop pole that yields the gain $\lambda_k$ to the origin and the negative real-axis reaches maximum at $k=n$, which can be denoted as $\overline{\psi}=\psi_n$. Thus, the minimum damping ratio is produced by $\hat{z}_{n,{\mathrm{fs}}}(s)$ as $\cos\overline{\psi}=\cos\psi_n$, whose value can be found using the algebraic relation in \eqref{eq:K-FS} and the geometric relation in Fig.~\ref{fig:rlocus-fs-mindamp}. More precisely, according to \eqref{eq:K-FS}, the closed-loop poles associated with the gain $\lambda_n$ should satisfy
\begin{align}\label{eq:K-FS-lambdan}
 \!\!\!\lambda_n=m|s|\left|  s+\dfrac{d+d_\mathrm{b}+d_\mathrm{t}}{m}\right|=m|s|^2=m\left(\dfrac{|\sigma_\mathrm{a}|}{\cos\overline{\psi}}\right)^2  \,.
\end{align}
Here, some geometric relations are used.
The second equality utilizes the geometric property that any point on the asymptotes is equidistance to the two open-loop poles $s_1$ and $s_2$, i.e.,  $|s|=|s+(d+d_\mathrm{b}+d_\mathrm{t})/m|$, which results from the fact that the orange and brown triangles in Fig.~\ref{fig:rlocus-fs-mindamp} are congruent by Side-Angle-Side rule. The last equality is due to the geometric definition in the brown triangle in Fig.~\ref{fig:rlocus-fs-mindamp}: $\cos\overline{\psi}$ is the ratio between the length of the side adjacent to $\overline{\psi}$ and the length of the hypotenuse, i.e., $\cos\overline{\psi}=|\sigma_\mathrm{a}|/|s|$. From \eqref{eq:K-FS-lambdan}, we can solve for the minimum damping ratio as
\begin{align*}
 \cos\overline{\psi}=|\sigma_\mathrm{a}|\sqrt{\dfrac{m}{\lambda_n}} = \left|  -\dfrac{d+d_\mathrm{b}+d_\mathrm{t}}{2m}\right| \sqrt{\dfrac{m}{\lambda_n}} =\dfrac{d+d_\mathrm{b}+d_\mathrm{t}}{2\sqrt{\lambda_n m}}\,,
\end{align*}
which concludes the proof of the first scenario in \eqref{eq:min-damp-fs}. 
\begin{figure}[t!]
\centering
\includegraphics[width=0.5\columnwidth]{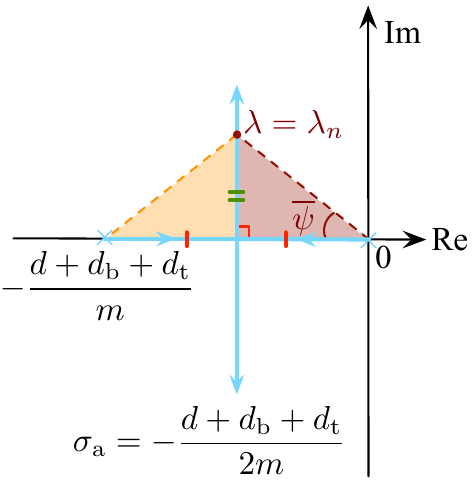}
\caption{Minimum damping ratio $\cos\overline{\psi}$ determined by $\lambda_n$.}\label{fig:rlocus-fs-mindamp}
\end{figure}

Next, we study $\underline{\alpha}$ by considering the two possibilities of where the gain equalling the second smallest eigenvalue $\lambda_2$ of the scaled Laplacian occurs on the locus, either on the asymptotes or not. Following from a similar argument on $\lambda_n$, the location where the gain $\lambda_2$ occurs depends on if 
\begin{align}\label{eq:osi-eig2}
    \lambda_2\geq m\sigma_\mathrm{a}^2=\dfrac{\left(d+d_\mathrm{b}+d_\mathrm{t}\right)^2}{4m}
\end{align}
or not. Particularly, if \eqref{eq:osi-eig2} holds, or equivalently $d_\mathrm{b}\leq 2\sqrt{\lambda_2m}-d-d_\mathrm{t}$, then the gain $\lambda_2$ occurs on the asymptotes; otherwise, the gain $\lambda_2$ occurs on the real-axis. We now discuss these two cases separately:
\begin{enumerate}
    \item The gain $\lambda_2$ occurs on the asymptotes ($d_\mathrm{b}\leq 2\sqrt{\lambda_2m}-d-d_\mathrm{t}$): As depicted in Fig.~\ref{fig:eig2-vertical}, this actually implies that all gains from $\lambda_2$ to $\lambda_n$ occur on the asymptotes, which further ensures that all poles of $\hat{z}_{k,{\mathrm{fs}}}(s)$, for $k\in\mathcal{N}\setminus\{1\}$, lie on the asymptotes with distances from the imaginary-axis being $|\sigma_\mathrm{a}|$. Thus, all $\hat{z}_{k,{\mathrm{fs}}}(s)$, for $k\in\mathcal{N}\setminus\{1\}$, trivially have the same decay rate
    \begin{align}     \label{eq:min-rate-v}   \underline{\alpha}=|\sigma_\mathrm{a}|=\dfrac{d+d_\mathrm{b}+d_\mathrm{t}}{2m}\,.
    \end{align}
    \item The gain $\lambda_2$ occurs on the real-axis ($d_\mathrm{b}> 2\sqrt{\lambda_2m}-d-d_\mathrm{t}$): As depicted in Fig.~\ref{fig:eig2-hori}, by symmetry, there are two real closed-loop poles that yield the gain $\lambda_2$, which are also the poles of $\hat{z}_{2,{\mathrm{fs}}}(s)$. Obviously, of those two poles, the one closer to the imaginary-axis gives the slower decay rate, denoted as $\min{\alpha_2}$, which is the minimum decay rate among all $\hat{z}_{k,{\mathrm{fs}}}(s)$ for $k\in\mathcal{N}\setminus\{1\}$ as well, i.e., $\underline{\alpha}=\min{\alpha_2}$, since the poles of those $\hat{z}_{k}(s)$ must lie in region $\text{Re}(s)\leq-\min{\alpha_2}$ based on the locus. Thus, the key is to compute the decay rate $\alpha_2$ of the poles at the gain $\lambda_2$, which can be easily solved by \eqref{eq:K-FS} if one notices that the poles are simply $s=-\alpha_2$ in this case. That is, we have    \begin{align}\label{eq:K-FS-lambda-2}
\lambda_2=&\ m|-\alpha_2|\left|  -\alpha_2+\dfrac{d+d_\mathrm{b}+d_\mathrm{t}}{m}\right|\nonumber\\=&\ m\alpha_2\left(\dfrac{d+d_\mathrm{b}+d_\mathrm{t}}{m}-\alpha_2\right)\nonumber\\=&\ (d+d_\mathrm{b}+d_\mathrm{t}) \alpha_2-m \alpha_2^2\,,
\end{align}
 where the second equality eliminates the absolute values by using the fact that, for the poles $s=-\alpha_2$ to lie on the real-axis segment of the locus, it must hold that $0=|s_1|<\alpha_2<|s_2|=(d+d_\mathrm{b}+d_\mathrm{t})/m$. Now, $\alpha_2$ can be solved from the univariate quadratic equation obtained from \eqref{eq:K-FS-lambda-2}, i.e.,
\begin{align}\label{eq:alpha-2-quad}
m \alpha_2^2-(d+d_\mathrm{b}+d_\mathrm{t}) \alpha_2+\lambda_2=0\,,
\end{align}
whose discriminant $(d+d_\mathrm{b}+d_\mathrm{t})^2-4m\lambda_2>0 $ since \eqref{eq:osi-eig2} does not hold in this case. Hence, \eqref{eq:alpha-2-quad} has two distinct real roots 
\begin{align*}    \alpha_2=\dfrac{d+d_\mathrm{b}+d_\mathrm{t}\pm\sqrt{(d+d_\mathrm{b}+d_\mathrm{t})^2-4m\lambda_2}}{2m}\,,
\end{align*}
of which the smaller one gives the minimum decay rate, i.e.,
\begin{align}
    \underline{\alpha}=\min{\alpha_2}=&\dfrac{d+d_\mathrm{b}+d_\mathrm{t}-\sqrt{(d+d_\mathrm{b}+d_\mathrm{t})^2-4m\lambda_2}}{2m}\nonumber\\   <&\dfrac{d+d_\mathrm{b}+d_\mathrm{t}}{2m}=|\sigma_\mathrm{a}|\,.\label{eq:min-rate-h}
\end{align}
\end{enumerate}

\begin{figure}[t!]
\centering
\subfigure[$\lambda_2$ occurs on the asymptotes]
{\includegraphics[width=0.49\columnwidth]{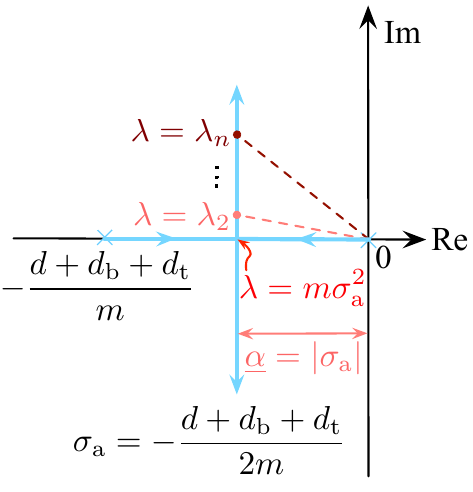}\label{fig:eig2-vertical}}
\subfigure[$\lambda_2$ occurs not on the asymptotes]
{\includegraphics[width=0.49\columnwidth]{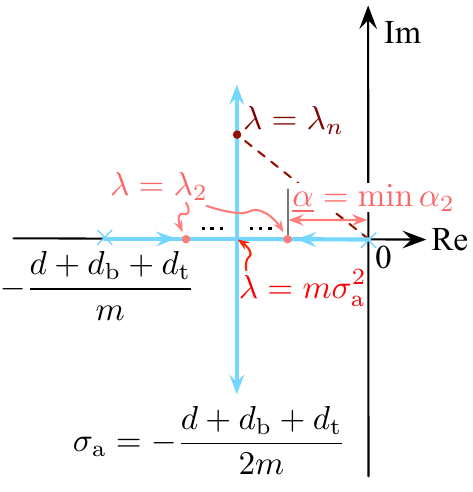}\label{fig:eig2-hori}}
\caption{Minimum decay rate $\underline{\alpha}$ determined by where $\lambda_2$ occurs.}
\label{fig:eig2}
\end{figure}
Therefore, we can combine \eqref{eq:min-rate-v} and \eqref{eq:min-rate-h} to yield \eqref{eq:min-decay-fs}. It is easy to confirm that both cases in \eqref{eq:min-decay-fs} may occur under assumption \eqref{eq:osi-eigmax} since the critical value in \eqref{eq:cri-d-damping} is naturally not less than the critical value in \eqref{eq:min-decay-fs}, i.e., $2\sqrt{\lambda_n m}-d-d_\mathrm{t}\geq 2\sqrt{\lambda_2m}-d-d_\mathrm{t}$.

Finally, we turn to the trivial scenario where $\lambda_n\leq m\sigma_\mathrm{a}^2$, i.e., $d_\mathrm{b}\geq 2\sqrt{\lambda_n m}-d-d_\mathrm{t}$. This, as mentioned before, characterizes the case where all gains from $\lambda_2$ to $\lambda_n$ occur along the real-axis segment as shown in Fig.~\ref{fig:rlocus-fs-maxeig-violate} so that power networks have no oscillatory behaviors. Observe from Fig.~\ref{fig:rlocus-fs-maxeig-violate} that $\overline{\psi}=\psi_k=0$ for any $k\in\mathcal{N}\setminus\{1\}$, which implies that $\cos\overline{\psi}=\cos0=1$. Notably, in this case, since the gain $\lambda_2$ occurs on the real-axis, the analysis yielding \eqref{eq:min-rate-h} still holds, which means that $\underline{\alpha}$ is given by the second case in \eqref{eq:min-decay-fs}. Again, it is easy to confirm that only the second case in \eqref{eq:min-decay-fs} may occur here since $2\sqrt{\lambda_n m}-d-d_\mathrm{t}\geq 2\sqrt{\lambda_2m}-d-d_\mathrm{t}$.

\subsection{Proof of Proposition~\ref{pro:sys-rate}}\label{app:sys-rate-pf}
Due to \eqref{eq:omega-t-decomple}, we begin by discussing frequency oscillations $\tilde {\boldsymbol{\omega}}_{\mathrm{fs}}(t)$ and COI frequency deviation $\bar{\omega}_{\mathrm{fs}}(t)$ separately.

It is well-known that the decay rate governs the exponential decay envelope of the transient response of a linear system~\cite[Module 5]{Driels1996linear}. Particularly, for an explicit expression of the step response of a linear two-pole system, analogous to $\hat{z}_{k,{\mathrm{fs}}}(s)$ for $k\in\mathcal{N}\setminus\{1\}$, we refer to the proof of~\cite[Lemma 3]{JIANG2024epsr}, where the role of decay rate is evident. Thus, by Theorem~\ref{thm:osc-limit-fs}, $\forall k\in\mathcal{N}\setminus\{1\}$, each $z_{\mathrm{u},k,{\mathrm{fs}}}(t)$ decays exponentially to $0$ with a rate of at least $\underline{\alpha}$ given by \eqref{eq:min-decay-fs}, i.e., 
\begin{align}\label{eq:rate-zuk}
    |z_{\mathrm{u},k,{\mathrm{fs}}}(t)|\leq \beta_k e^{-\underline{\alpha}t}\,,\qquad \forall t \geq 0,k\in\mathcal{N}\setminus\{1\}\,,
\end{align}
for some constant $\beta_k>0$.
Now, note from \eqref{eq:w-til-t} that 
\begin{align}\label{tilde-linear}
    \tilde {\boldsymbol{\omega}}_{\mathrm{fs}}(t)=\sum_{k=2}^n z_{\mathrm{u},k,{\mathrm{fs}}}(t)\boldsymbol{\mu}_k
\end{align}
with $\boldsymbol{\mu}_k:=\boldsymbol{R}^{-\frac{1}{2}} 
    \boldsymbol{v}_k\boldsymbol{v}_k^T\boldsymbol{R}^{-\frac{1}{2}}\boldsymbol{u}_0$ to streamline the notation. Thus, it follows from \eqref{tilde-linear} that
\begin{align}\label{tilde-linear-bound}
    \|\tilde {\boldsymbol{\omega}}_{\mathrm{fs}}(t)\|\leq&\sum_{k=2}^n |z_{\mathrm{u},k,{\mathrm{fs}}}(t)|\|\boldsymbol{\mu}_k\|\leq\sum_{k=2}^n \beta_k e^{-\underline{\alpha}t}\|\boldsymbol{\mu}_k\|\nonumber\\=&\left(\sum_{k=2}^n \beta_k \|\boldsymbol{\mu}_k\|\right)e^{-\underline{\alpha}t}=\delta_1 e^{-\underline{\alpha}t}
\end{align}
with $\delta_1:=\sum_{k=2}^n \beta_k \|\boldsymbol{\mu}_k\|$, where the first inequality is due to the triangle inequality and the second inequality uses \eqref{eq:rate-zuk}. Therefore, $\tilde {\boldsymbol{\omega}}_{\mathrm{fs}}(t)$ decays exponentially to zero with a rate of at least $\underline{\alpha}$ as well. This concludes the proof of the statement on $\tilde {\boldsymbol{\omega}}_{\mathrm{fs}}(t)$.

As for the main mode, note that $\hat{z}_{1,{\mathrm{fs}}}(s)$ in \eqref{eq:z1fs} has only one pole at $-(d+d_\mathrm{b}+d_\mathrm{t})/m$, which actually lies at the leftmost position on the root locus in Fig.~\ref{fig:rlocus-fs}. Thus, $z_{\mathrm{u},1,{\mathrm{fs}}}(t)$ must achieve a larger decay rate $(d+d_\mathrm{b}+d_\mathrm{t})/m$ than any $z_{\mathrm{u},k,{\mathrm{fs}}}(t)$ with $k\in\mathcal{N}\setminus\{1\}$, which ensures that $(d+d_\mathrm{b}+d_\mathrm{t})/m>\underline{\alpha}$. This can be easily confirmed by examining the explicit expression of $z_{\mathrm{u},1,{\mathrm{fs}}}(t)$, i.e., 
\begin{align}
z_{\mathrm{u},1,{\mathrm{fs}}}(t)\!:=& \mathscr{L}^{-1}\!\left\{\dfrac{\hat{z}_{1,{\mathrm{fs}}}(s)}{s}\right\} \!\!=\!\mathscr{L}^{-1}\!\left\{\dfrac{1}{s\left(m s + d+d_\mathrm{b} + d_{\mathrm{t}}\right)}\right\} \nonumber\\=&\mathscr{L}^{-1}\!\left\{\dfrac{1}{d+d_\mathrm{b} + d_{\mathrm{t}}}\!\left[\dfrac{1}{s}\!-\!\frac{1}{s+(d+d_\mathrm{b} + d_{\mathrm{t}})/m}\right]\right\}\nonumber\\=&\dfrac{1}{d+d_\mathrm{b} + d_{\mathrm{t}}}\left(1-e^{-\frac{d+d_\mathrm{b} + d_{\mathrm{t}}}{m}t}\right)\,,\label{eq:zu1fs}
\end{align}
which converges to the steady-state value $1/(d+d_\mathrm{b} + d_{\mathrm{t}})$ exponentially fast with a rate $(d+d_\mathrm{b}+d_\mathrm{t})/m$. According to \eqref{eq:w-bar-t}, substituting \eqref{eq:zu1fs} to 
\begin{align*}
    \bar{\omega}_{\mathrm{fs}}(t):=\dfrac{\sum_{i=1}^n u_{0,i}}{\sum_{i=1}^n r_i}z_{\mathrm{u},1,{\mathrm{fs}}}(t)
\end{align*}
yields \eqref{eq:coi-fs-inf} and \eqref{eq:coi-fs-t}, which concludes the proof of the statement on $\bar{\omega}_{\mathrm{fs}}(t)$.

Now, we are ready to characterize the convergence rate of frequency deviations $\boldsymbol{\omega}_{\mathrm{fs}}(t)$ of the system as
\begin{align}
    \|\boldsymbol{\omega}_{\mathrm{fs}}(t)-\bar{\omega}_{\mathrm{fs}}(\infty)\mathbbold{1}_n\|\!=&\ \!\|\bar{\omega}_{\mathrm{fs}}(t)\mathbbold{1}_n+\tilde {\boldsymbol{\omega}}_{\mathrm{fs}}(t)-\bar{\omega}_{\mathrm{fs}}(\infty)\mathbbold{1}_n\|\nonumber\\=&\ \!\|\left(\bar{\omega}_{\mathrm{fs}}(t)-\bar{\omega}_{\mathrm{fs}}(\infty)\right)\mathbbold{1}_n+\tilde {\boldsymbol{\omega}}_{\mathrm{fs}}(t)\|\nonumber\\
    \stackrel{\circled{1}}{=}&\ \!\|-\bar{\omega}_{\mathrm{fs}}(\infty)e^{-\frac{d+d_\mathrm{b} + d_{\mathrm{t}}}{m}t}\mathbbold{1}_n+\tilde {\boldsymbol{\omega}}_{\mathrm{fs}}(t)\|\nonumber\\
    \stackrel{\circled{2}}{\leq}&\ \!\|-\bar{\omega}_{\mathrm{fs}}(\infty)e^{-\frac{d+d_\mathrm{b} + d_{\mathrm{t}}}{m}t}\mathbbold{1}_n\|\!+\!\|\tilde {\boldsymbol{\omega}}_{\mathrm{fs}}(t)\|\nonumber
    \\
    \stackrel{\circled{3}}{\leq}&\ \!\|\bar{\omega}_{\mathrm{fs}}(\infty)\mathbbold{1}_n\|e^{-\frac{d+d_\mathrm{b} + d_{\mathrm{t}}}{m}t}+\delta_1 e^{-\underline{\alpha}t}\nonumber
    \\
    \stackrel{\circled{4}}{<}&\ \!\|\bar{\omega}_{\mathrm{fs}}(\infty)\mathbbold{1}_n\|e^{-\underline{\alpha}t}+\delta_1 e^{-\underline{\alpha}t}\nonumber
    \\
    =&\left(\|\bar{\omega}_{\mathrm{fs}}(\infty)\mathbbold{1}_n\|+\delta_1 \right)e^{-\underline{\alpha}t}\,,\nonumber
\end{align}
where \circled{1} uses \eqref{eq:coi-fs-t}, \circled{2} is due to the triangle inequality, \circled{3} uses \eqref{tilde-linear-bound}, \circled{4} is due to the fact that $(d+d_\mathrm{b}+d_\mathrm{t})/m>\underline{\alpha}$. This concludes the proof of \eqref{eq:exp-wfs} with $\delta:=\|\bar{\omega}_{\mathrm{fs}}(\infty)\mathbbold{1}_n\|+\delta_1$.

\subsection{Proof of Lemma~\ref{lem:rlocus-VI}} \label{app:lem-rlocus-VI}
Applying \eqref{eq:go} and $\hat{c}_\mathrm{o}(s)=\hat{c}_{\mathrm{vi}}(s)$ in \eqref{eq:dy-vi} to \eqref{eq:zk-s} for $k=1$ shows that \begin{align}\label{eq:z1vi}
    \hat{z}_{1,{\mathrm{vi}}}(s)=&\ \dfrac{ \tau s + 1}{(m\!+\!m_{\mathrm{v}})\tau s^2 \!+\! \left[m\!+\!m_{\mathrm{v}} \!+\! (d\!+\!d_\mathrm{b}) \tau \right] s \!+\! d\!+\!d_\mathrm{b} \!+\! d_{\mathrm{t}}}\nonumber\\=&\ \dfrac{s + \tau^{-1}}{(m+m_{\mathrm{v}})\left( s^2 + 2\xi\omega_\mathrm{n} s + \omega_\mathrm{n}^2 \right)}
\end{align}
with $\omega_\mathrm{n}$ and $\xi$ given by \eqref{eq:ex-xi-wn}, which directly gives $\hat{L}_{\mathrm{vi}}(s)$ in \eqref{eq:Loopgain-vi}. Now, the root locus of the system with the loop-gain $\hat{L}_{\mathrm{vi}}(s)$ can be sketched using the standard rules~\cite[Module 10]{Driels1996linear}, which requires us to first investigate the open-loop poles and zeros. 
Clearly, \eqref{eq:Loopgain-vi} has an open-loop zero at $-\tau^{-1}$, an open-loop pole at $s_1=0$, and two other open-loop poles coming from roots of $s^2 + 2\xi\omega_\mathrm{n} s + \omega_\mathrm{n}^2=0$. Specifically, for Nadir elimination tuning where $(d_{\mathrm{b}}, m_{\mathrm{v}})$ satisfies \eqref{eq:mvmin}, we have either $\xi=1$ or $\xi>1$ by the proof of~\cite[Theorem 4]{jiang2021tac} as well as $\xi\omega_\mathrm{n} \leq \tau^{-1}$ in both cases by the proof of~\cite[Theorem 6]{jiang2021tac}. We now discuss these two cases separately borrowing some preliminary analysis in~\cite{jiang2021tac,jiangtps2021}:
\begin{enumerate}
    \item $\xi=1$: There are two repeated negative real roots, i.e.,
\begin{align*}
    s_2=s_3=-\omega_\mathrm{n}=-\sqrt{\cfrac{d+d_{\mathrm{b}} +d_{\mathrm{t}}}{(m+m_{\mathrm{v}})\tau}}>-\tau^{-1}\,,
\end{align*}
where the inequality is due to $\omega_\mathrm{n} =\xi\omega_\mathrm{n} \leq \tau^{-1} $ for $\xi=1$ and the possibility that $\omega_\mathrm{n} = \tau^{-1}$ can be ruled out by contradiction. To see this, we note that $\omega_\mathrm{n} = \tau^{-1}$ implies that $m+m_{\mathrm{v}}=\tau(d+d_{\mathrm{b}} +d_{\mathrm{t}})<\tau\left(\!\sqrt{d_\mathrm{t}}+\sqrt{d+d_\mathrm{t}+d_{\mathrm{b}}}\right)^2$ via simple algebraic manipulation, which contradicts the Nadir elimination condition \eqref{eq:mvmin}. Then, the real-axis segments for $-\tau^{-1}\leq s\leq s_2=s_3$ and $s_2=s_3\leq s\leq s_1$ form part of the locus since they are to the left of an odd number of poles or zeros. Particularly, while one branch goes from one of the repeated open-loop poles to the zero at $-\tau^{-1}$ along a real-axis segment as $\lambda$ increases, two branches of the locus go from the other open-loop poles to infinity since $N_\mathrm{p}-N_\mathrm{z}=3-1=2$.
To determine how the poles go to infinity as $\lambda$ increases, we next compute the center and orientation of the asymptotes. Precisely, the asymptotes intersect the real-axis at
\begin{align}\label{eq:sigma-vi-cri}
    \sigma_\mathrm{a} 
    =&\ \frac{\text{Re}(s_1)+\text{Re}(s_2)+\text{Re}(s_3)-\text{Re}(-\tau^{-1})}{N_\mathrm{p}-N_\mathrm{z}}\nonumber\\
    =&\  \frac{0+(-\omega_\mathrm{n})+(-\omega_\mathrm{n})-(-\tau^{-1})}{2}
    \nonumber\\=&\ \dfrac{\tau^{-1}}{2}-\omega_\mathrm{n}>-\omega_\mathrm{n}
\end{align}
and incline to the positive real-axis at
\begin{align}\label{eq:phase-sigma}
\dfrac{180\degree(1+2l)}{N_\mathrm{p}-N_\mathrm{z}}\quad\text{with }l=0,1,\ldots, (N_\mathrm{p}-N_\mathrm{z}-1) \,, 
\end{align}
which yields $90\degree$ and $270\degree$ here for $N_\mathrm{p}-N_\mathrm{z}=2$. To better understand $\sigma_\mathrm{a}$, we can further express it in terms of parameters by applying relations in \eqref{eq:ex-xi-wn} to \eqref{eq:sigma-vi-cri}, i.e.,
\begin{align}\label{eq:sigma-vi-cri0}
 \sigma_\mathrm{a} =&\ \dfrac{\tau^{-1}}{2}\!-\dfrac{\tau^{-1}\!+(d+d_{\mathrm{b}})/(m+m_{\mathrm{v}})}{2\xi}\nonumber\\=&\ \dfrac{\tau^{-1}}{2}\!-\dfrac{\tau^{-1}\!+(d+d_{\mathrm{b}})/(m+m_{\mathrm{v}})}{2}\nonumber\\=&-\dfrac{d+d_{\mathrm{b}}}{2(m+m_{\mathrm{v}})}\!<0\,,
\end{align}
where the second equality uses $\xi=1$. Combining \eqref{eq:sigma-vi-cri} and \eqref{eq:sigma-vi-cri0}, we know that $-\omega_\mathrm{n}<\sigma_\mathrm{a}<0$. Now, we are able to sketch real-axis segments and asymptotes of locus as in Fig.~\ref{fig:rlocus-vi-cri}. Based on our analysis before, one branch goes from a pole at $-\omega_\mathrm{n}$ to the zero at $-\tau^{-1}$ along the real-axis, while two branches originate from the other two poles, approach each other along the real-axis until a break point, and branch off to infinity along asymptotes. Although it is difficult to analytically solve the break point in this scenario due to the order of $\hat{L}_{\mathrm{vi}}(s)$, we still know the branches must approach the asymptotes from the left and thus the break point is naturally to the left of $\sigma_\mathrm{a}$ as sketched in Fig.~\ref{fig:rlocus-vi-cri}. We now illustrate why the branches approach the asymptotes from the left. By~\cite[Module 10]{Driels1996linear}, the sum of the real parts of the closed-loop poles on locus is constant when $N_\mathrm{p}-N_\mathrm{z}=2$, independent of $\lambda$. Thus, given that a closed-loop pole moves left to the zero, the other two poles after breaking away from the real-axis must move right toward asymptotes to keep the sum of real parts of closed-loop poles constant.

\item $\xi>1$: There are two distinct negative real roots, i.e.,
\begin{align*}
    s_2 = -\omega_\mathrm{n}\!\left(\xi- \!\sqrt{\xi^2-\!1}\right)>s_3 =-\omega_\mathrm{n}\!\left(\xi+ \!\sqrt{\xi^2-\!1}\right)\,.
\end{align*}
To determine the relative position of the zero $-\tau^{-1}$ with respect to the poles on the real-axis, we need to analyze the sign of the following difference 
\begin{align}\label{eq:diff-zero-pole}
   \!\! \tau^{-1}\!-\omega_\mathrm{n}\!\left(\xi\!+ \!\sqrt{\xi^2\!-\!1}\right)\!=\!\left(\tau^{-1}\!-\!\xi\omega_\mathrm{n}\right)\!-\!\omega_\mathrm{n}\sqrt{\xi^2\!-\!1}\,.
\end{align}
We next show that \eqref{eq:diff-zero-pole} is positive and thus $-\tau^{-1}$ lies to the left of $s_3$ by examining the value of $\omega_\mathrm{n}\sqrt{\xi^2-\!1}$. First, since $\omega_\mathrm{n}>0$ and $\xi>1$, it must hold that $\omega_\mathrm{n}\sqrt{\xi^2-\!1}>0$.
Moreover, we can use \eqref{eq:ex-xi-wn} to rewrite $\omega_\mathrm{n}\sqrt{\xi^2-\!1}$ as 
\begin{align}\label{eq:bound-secondpart}
    &\omega_\mathrm{n}\sqrt{\xi^2-\!1}\nonumber\\&=\sqrt{(\xi\omega_\mathrm{n})^2-\!\omega_\mathrm{n}^2}\nonumber\\&=\sqrt{\left(\dfrac{\tau^{-1}+(d+d_{\mathrm{b}})/(m+m_{\mathrm{v}})}{2}\right)^2\!-\cfrac{d+d_{\mathrm{b}} +d_{\mathrm{t}}}{(m+m_{\mathrm{v}})\tau}}\nonumber
    \\&=\sqrt{\left(\dfrac{\tau^{-1}-(d+d_{\mathrm{b}})/(m+m_{\mathrm{v}})}{2}\right)^2\!-\!\cfrac{d_{\mathrm{t}}}{(m+m_{\mathrm{v}})\tau}}\nonumber\\&=\sqrt{\left(\tau^{-1}\!-\!\dfrac{\tau^{-1}+(d+d_{\mathrm{b}})/(m+m_{\mathrm{v}})}{2}\right)^2\!\!-\!\cfrac{d_{\mathrm{t}}}{(m+m_{\mathrm{v}})\tau}}
    \nonumber\\&=\sqrt{\left(\tau^{-1}-\xi\omega_\mathrm{n}\right)^2-\!\cfrac{d_{\mathrm{t}}}{(m+m_{\mathrm{v}})\tau}}<\tau^{-1}-\xi\omega_\mathrm{n}\,,
\end{align}
which implies that \eqref{eq:diff-zero-pole} is positive, i.e., $\tau^{-1}+s_3>0$. This shows that the zero is to the left of $s_3$ since $-\tau^{-1}<s_3$.
Then, the real-axis segments for $-\tau^{-1}\leq s\leq s_3$ and $s_2\leq s\leq s_1$ form part of the locus since they are to the left of an odd number of poles or zeros. Particularly, while one branch goes from the open-loop pole $s_3$ to the zero at $-\tau^{-1}$ along a real-axis segment as $\lambda$ increases, two branches of the locus go from the other open-loop poles to infinity since $N_\mathrm{p}-N_\mathrm{z}=3-1=2$. Similarly, the asymptotes intersect the real-axis at
\begin{align}\label{eq:sigma-vi-over}
    \sigma_\mathrm{a} 
    =&\ \frac{\text{Re}(s_1)+\text{Re}(s_2)+\text{Re}(s_3)-\text{Re}(-\tau^{-1})}{N_\mathrm{p}-N_\mathrm{z}}\nonumber\\
    =&\  \frac{0\!-\omega_\mathrm{n}\!\left(\xi- \!\sqrt{\xi^2-\!1}\right)\!-\omega_\mathrm{n}\!\left(\xi\!+ \!\sqrt{\xi^2-\!1}\right)\!-\!(-\tau^{-1})}{2}
    \nonumber\\=&\ \dfrac{\tau^{-1}}{2}-\xi\omega_\mathrm{n}
\end{align}
and incline to the positive real-axis at $90\degree$ and $270\degree$ due to \eqref{eq:phase-sigma}. Again, we can express $\sigma_\mathrm{a}$ in terms of parameters by applying \eqref{eq:ex-xi-wn} to \eqref{eq:sigma-vi-over}, i.e.,
\begin{align}\label{eq:sigma-vi-over0}
 \sigma_\mathrm{a} =&\ \dfrac{\tau^{-1}}{2}\!-\dfrac{\tau^{-1}\!+(d+d_{\mathrm{b}})/(m+m_{\mathrm{v}})}{2}\nonumber\\=&-\dfrac{d+d_{\mathrm{b}}}{2(m+m_{\mathrm{v}})}\!<0\,.
\end{align}
To determine the relative position of $\sigma_\mathrm{a}$ with respect to the pole $s_2$, we analyze the sign of the following difference
\begin{align}\label{eq:sigma-vi-over-s2}
    \sigma_\mathrm{a}-s_2\stackrel{\circled{1}}{=}\ &\dfrac{\tau^{-1}}{2}-\xi\omega_\mathrm{n}-\left[-\omega_\mathrm{n}\!\left(\xi- \!\sqrt{\xi^2-\!1}\right)\right]\nonumber\\=\ &\dfrac{\tau^{-1}}{2}-\omega_\mathrm{n}\sqrt{\xi^2-\!1}\nonumber\\\stackrel{\circled{2}}{>}\ &\dfrac{\tau^{-1}}{2}-\left(\tau^{-1}-\xi\omega_\mathrm{n}\right)\nonumber\\=\ &\xi\omega_\mathrm{n}-\dfrac{\tau^{-1}}{2}\stackrel{\circled{3}}{=}-\sigma_\mathrm{a}>0\,,
\end{align}
where \circled{1} uses \eqref{eq:sigma-vi-over}, \circled{2} is due to the inequality in \eqref{eq:bound-secondpart}, and \circled{3} uses \eqref{eq:sigma-vi-over} again. From \eqref{eq:sigma-vi-over0} and \eqref{eq:sigma-vi-over-s2}, we know that
$s_2<s_2/2<\sigma_\mathrm{a}<0$. Now, we are able to sketch real-axis segments and asymptotes of locus as in Fig.~\ref{fig:rlocus-vi-over}, where the remaining locus are also sketched following a similar argument as in the $\xi=1$ case.
\end{enumerate}

\subsection{Proof of Proposition~\ref{pro:sys-rate-vi}}\label{app:sys-rate-vi-pf}
Due to \eqref{eq:omega-t-decomple}, we discuss frequency oscillations $\tilde {\boldsymbol{\omega}}_{\mathrm{vi}}(t)$ and COI frequency deviation $\bar{\omega}_{\mathrm{vi}}(t)$ separately.    

To bound decay rate of $\tilde {\boldsymbol{\omega}}_{\mathrm{vi}}(t)$, we investigate $\hat{z}_{k,{\mathrm{vi}}}(s)$ for $k\in\mathcal{N}\setminus\{1\}$. Applying \eqref{eq:go} and $\hat{c}_\mathrm{o}(s)=\hat{c}_{\mathrm{vi}}(s)$ in \eqref{eq:dy-vi} to \eqref{eq:zk-s} shows that 
\begin{align}\label{eq:zkvi}
\hat{z}_{k,{\mathrm{vi}}}(s)=\dfrac{s(s + \tau^{-1})}{(m\!+\!m_{\mathrm{v}})s\left( s^2 \!+\! 2\xi\omega_\mathrm{n} s \!+\! \omega_\mathrm{n}^2 \right)\!+\!\lambda_k(s \!+\! \tau^{-1})}   
\end{align}
with $\omega_\mathrm{n}$ and $\xi$ given by \eqref{eq:ex-xi-wn}. Applying the final-value theorem to the unit-step response of \eqref{eq:zkvi} yields
\begin{align*}
    z_{\mathrm{u},k,{\mathrm{vi}}}(\infty)\!:= \lim_{s\to 0}s\dfrac{\hat{z}_{k,{\mathrm{vi}}}(s)}{s}=\hat{z}_{k,{\mathrm{vi}}}(0)=0\,,
\end{align*}
which shows that each $z_{\mathrm{u},k,{\mathrm{vi}}}(t)$ eventually decays to $0$. Although, due to the high-order of each $\hat{z}_{k,{\mathrm{vi}}}(s)$, it is too cumbersome to provide an explicit expression for $z_{\mathrm{u},k,{\mathrm{vi}}}(t)$, the fact that the exponential envelope of $z_{\mathrm{u},k,{\mathrm{vi}}}(t)$ decays with a rate determined by pole locations of $\hat{z}_{k,{\mathrm{vi}}}(s)$ remains true. Thus, we now study the decay rate through Fig.~\ref{fig:rlocus-vi}. Observe from Fig.~\ref{fig:rlocus-vi-cri} and Fig.~\ref{fig:rlocus-vi-over} that, $\forall k\in\mathcal{N}\setminus\{1\}$, each $\hat{z}_{k,{\mathrm{vi}}}(s)$ must have two poles located within the vertical strip $-\omega_\mathrm{n}<\text{Re}(s)<0$ if $\xi=1$ and $-\omega_\mathrm{n}(\xi- \!\sqrt{\xi^2-\!1})<\text{Re}(s)<0$ if $\xi>1$. We claim that
    \begin{align}\label{eq:compare-wn-mpole}
     \omega_\mathrm{n}>\omega_\mathrm{n}\left(\xi\!-\!\sqrt{\xi^2\!-\!1}\right)\,,\qquad\forall \xi>1\,, 
    \end{align}
    Thus, $\forall k\in\mathcal{N}\setminus\{1\}$, each $\hat{z}_{k,{\mathrm{vi}}}(s)$ has two poles located within the vertical strip $-\omega_\mathrm{n}<\text{Re}(s)<0$ for sure, regardless of whether $\xi=1$ or $\xi>1$. To see why \eqref{eq:compare-wn-mpole} is the case, we first note that $\xi\!+ \!\sqrt{\xi^2\!-\!1}>\xi>1$
    holds trivially, which implies that 
    \begin{align}
        1>&\ \dfrac{1}{\xi\!+\!\sqrt{\xi^2\!-\!1}}=\dfrac{\xi\!-\!\sqrt{\xi^2\!-\!1}}{\left(\xi\!+\!\sqrt{\xi^2\!-\!1}\right)\left(\xi\!-\!\sqrt{\xi^2\!-\!1}\right)}\nonumber\\=&\ \xi\!-\!\sqrt{\xi^2\!-\!1}\,.\label{eq:xi-nega-1}
    \end{align}
    Since $\omega_\mathrm{n}>0$, \eqref{eq:compare-wn-mpole} follows from \eqref{eq:xi-nega-1} directly. Therefore, $\forall k\in\mathcal{N}\setminus\{1\}$, each $z_{\mathrm{u},k,{\mathrm{vi}}}(t)$ must decay exponentially to $0$ with a rate smaller than $\omega_\mathrm{n}$, i.e., 
\begin{align*}
    |z_{\mathrm{u},k,{\mathrm{vi}}}(t)|\leq \gamma_k e^{-\rho_1 t}\,,\qquad \forall t \geq 0,k\in\mathcal{N}\setminus\{1\}\,,
\end{align*}
for some constant $\gamma_k>0$ and $0<\rho_1<\omega_\mathrm{n}$, while
\begin{align}
    \forall \kappa>\omega_\mathrm{n}\,,\quad\limsup_{t\to\infty} |z_{\mathrm{u},k,{\mathrm{vi}}}(t)|e^{\kappa t}=\infty\,.
\end{align}
Then, it follows from a similar argument as \eqref{tilde-linear-bound} in the proof of Proposition~\ref{pro:sys-rate} that $\tilde {\boldsymbol{\omega}}_{\mathrm{vi}}(t)=\sum_{k=2}^n z_{\mathrm{u},k,{\mathrm{vi}}}(t)\boldsymbol{\mu}_k$ decays exponentially to zero with a rate smaller than $\omega_\mathrm{n}$, i.e.,
\begin{align}\label{eq:tilde-linear-bound-vi}
    \|\tilde {\boldsymbol{\omega}}_{\mathrm{vi}}(t)\|\leq\sigma_1 e^{-\rho_1 t}\,,\qquad\exists \rho_1\in(0,\omega_\mathrm{n})\,,
\end{align}
with $\sigma_1:=\sum_{k=2}^n \gamma_k \|\boldsymbol{\mu}_k\|$. 
 
As for the main mode, note that $\hat{z}_{1,{\mathrm{vi}}}(s)$ in \eqref{eq:z1vi} has either two repeated poles at $-\omega_\mathrm{n}$ if $\xi=1$ or two distinct poles at $-\omega_\mathrm{n}(\xi\pm \!\sqrt{\xi^2-\!1})$ if $\xi>1$, which, combined with \eqref{eq:compare-wn-mpole}, indicates that at least one pole of $\hat{z}_{1,{\mathrm{vi}}}(s)$ is within the vertical strip $-\omega_\mathrm{n}\leq\text{Re}(s)<0$ in Fig.~\ref{fig:rlocus-vi}. Thus, $z_{\mathrm{u},1,{\mathrm{vi}}}(t)$ must exhibit a decay rate not exceeding $\omega_\mathrm{n}$. This can be easily confirmed by examining the explicit expression of $z_{\mathrm{u},1,{\mathrm{vi}}}(t):=\mathscr{L}^{-1}\{\hat{z}_{1,{\mathrm{vi}}}(s)/s\}$ which can be directly inferred from the unit-step response of a
second-order system with a zero provided in the proof of~\cite[Theorem 4]{jiang2021tac}. According to the step response therein, $z_{\mathrm{u},1,{\mathrm{vi}}}(t)$ converges to the steady-state value
\begin{align*}
    z_{\mathrm{u},1,{\mathrm{vi}}}(\infty)=\dfrac{\tau^{-1}}{(m\!+\!m_{\mathrm{v}})\omega_\mathrm{n}^2}
    =\dfrac{1}{d+d_{\mathrm{b}} +d_{\mathrm{t}}}
\end{align*}
exponentially fast with a rate not exceeding $\omega_\mathrm{n}$, where the second equality uses \eqref{eq:wn}. Thus, from \eqref{eq:w-bar-t}, we know that the COI frequency deviation
\begin{align*}
    \bar{\omega}_{\mathrm{vi}}(t):=\dfrac{\sum_{i=1}^n u_{0,i}}{\sum_{i=1}^n r_i}z_{\mathrm{u},1,{\mathrm{vi}}}(t)
\end{align*}
converges to the steady-state value $\bar{\omega}_{\mathrm{vi}}(\infty)$ in \eqref{eq:coi-vi-inf} exponentially fast with a rate not exceeding $\omega_\mathrm{n}$ as well, i.e.,
\begin{align}\label{eq:osci-vi-bound}
    | \bar{\omega}_{\mathrm{vi}}(t)-\bar{\omega}_{\mathrm{vi}}(\infty)|\leq\sigma_2 e^{-\rho_2 t}
\end{align}
for some constant $\sigma_2>0$ and $0<\rho_2\leq\omega_\mathrm{n}$.

Now, we are ready to characterize the convergence rate of frequency deviations $\boldsymbol{\omega}_{\mathrm{vi}}(t)$ of the system as
\begin{align}
    \|\boldsymbol{\omega}_{\mathrm{vi}}(t)-\bar{\omega}_{\mathrm{vi}}(\infty)\mathbbold{1}_n\|\!=&\ \!\|\bar{\omega}_{\mathrm{vi}}(t)\mathbbold{1}_n+\tilde {\boldsymbol{\omega}}_{\mathrm{vi}}(t)-\bar{\omega}_{\mathrm{vi}}(\infty)\mathbbold{1}_n\|\nonumber\\=&\ \!\|\left(\bar{\omega}_{\mathrm{vi}}(t)-\bar{\omega}_{\mathrm{vi}}(\infty)\right)\mathbbold{1}_n+\tilde {\boldsymbol{\omega}}_{\mathrm{vi}}(t)\|\nonumber\\
    \stackrel{\circled{1}}{\leq}&\ \!| \bar{\omega}_{\mathrm{vi}}(t)-\bar{\omega}_{\mathrm{vi}}(\infty)|\|\mathbbold{1}_n\|\!+\!\|\tilde {\boldsymbol{\omega}}_{\mathrm{vi}}(t)\|\nonumber
    \\
    \stackrel{\circled{2}}{\leq}&\ \!\sigma_2 e^{-\rho_2 t}\|\mathbbold{1}_n\| +\sigma_1 e^{-\rho_1 t}\nonumber
    \\
    \leq&\ \!\left(\sigma_2 \|\mathbbold{1}_n\|+\sigma_1 \right)e^{-\min(\rho_1,\rho_2) t}\,,\nonumber
\end{align}
where \circled{1} is due to the triangle inequality and \circled{2} uses \eqref{eq:osci-vi-bound} and \eqref{eq:tilde-linear-bound-vi}. This concludes the proof of \eqref{eq:exp-wvi} with $\sigma:=\sigma_2 \|\mathbbold{1}_n\|+\sigma_1$ and $\rho:=\min(\rho_1,\rho_2)$.


\section*{Acknowledgment}
The authors would like to acknowledge and thank Steven Low, Fernando Paganini, and Enrique Mallada for their insightful comments that helped improve this manuscript.


\bibliographystyle{IEEEtran}
\bibliography{main}

\begin{thebibliography}{10}
\providecommand{\url}[1]{#1}
\csname url@samestyle\endcsname
\providecommand{\newblock}{\relax}
\providecommand{\bibinfo}[2]{#2}
\providecommand{\BIBentrySTDinterwordspacing}{\spaceskip=0pt\relax}
\providecommand{\BIBentryALTinterwordstretchfactor}{4}
\providecommand{\BIBentryALTinterwordspacing}{\spaceskip=\fontdimen2\font plus
\BIBentryALTinterwordstretchfactor\fontdimen3\font minus \fontdimen4\font\relax}
\providecommand{\BIBforeignlanguage}[2]{{%
\expandafter\ifx\csname l@#1\endcsname\relax
\typeout{** WARNING: IEEEtran.bst: No hyphenation pattern has been}%
\typeout{** loaded for the language `#1'. Using the pattern for}%
\typeout{** the default language instead.}%
\else
\language=\csname l@#1\endcsname
\fi
#2}}
\providecommand{\BIBdecl}{\relax}
\BIBdecl

\bibitem{milano2018}
F.~Milano, F.~D\"orfler, G.~Hug, D.~J. Hill, and G.~Verbi{\v c}, ``Foundations and challenges of low-inertia systems (invited paper),'' in \emph{Proc. of Power Systems Computation Conference}, June 2018, pp. 1--25.

\bibitem{milano2017rotor}
F.~Milano, ``Rotor speed-free estimation of the frequency of the center of inertia,'' \emph{IEEE Transactions on Power Systems}, vol.~33, no.~1, pp. 1153--1155, Jan. 2018.

\bibitem{azizi2020local}
S.~Azizi, M.~Sun, G.~Liu, and V.~Terzija, ``Local frequency-based estimation of the rate of change of frequency of the center of inertia,'' \emph{IEEE Transactions on Power Systems}, vol.~35, no.~6, pp. 4948--4951, Nov. 2020.

\bibitem{Min2021lcss}
H.~Min, F.~Paganini, and E.~Mallada, ``Accurate reduced-order models for heterogeneous coherent generators,'' \emph{IEEE Control Systems Letters}, vol.~5, no.~5, pp. 1741--1746, Nov. 2021.

\bibitem{jiangtps2021}
Y.~Jiang, E.~Cohn, P.~Vorobev, and E.~Mallada, ``Storage-based frequency shaping control,'' \emph{IEEE Transactions on Power Systems}, vol.~36, no.~6, pp. 5006--5019, Nov. 2021.

\bibitem{jiang2021lcss}
Y.~Jiang, A.~Bernstein, P.~Vorobev, and E.~Mallada, ``Grid-forming frequency shaping control for low-inertia power systems,'' \emph{IEEE Control Systems Letters}, vol.~5, no.~6, pp. 1988--1993, Dec. 2021.

\bibitem{arani2012implementing}
M.~F.~M. Arani and E.~F. El-Saadany, ``Implementing virtual inertia in {DFIG}-based wind power generation,'' \emph{IEEE Transactions on Power Systems}, vol.~28, no.~2, pp. 1373--1384, May 2013.

\bibitem{fang2017distributed}
J.~Fang, H.~Li, Y.~Tang, and F.~Blaabjerg, ``Distributed power system virtual inertia implemented by grid-connected power converters,'' \emph{IEEE Transactions on Power Electronics}, vol.~33, no.~10, pp. 8488--8499, Oct. 2018.

\bibitem{jiang2021tac}
Y.~Jiang, R.~Pates, and E.~Mallada, ``Dynamic droop control in low-inertia power systems,'' \emph{IEEE Transactions on Automatic Control}, vol.~66, no.~8, pp. 3518--3533, Aug. 2021.

\bibitem{Kundur1991tps}
M.~Klein, G.~Rogers, and P.~Kundur, ``A fundamental study of inter-area oscillations in power systems,'' \emph{IEEE Transactions on Power Systems}, vol.~6, no.~3, pp. 914--921, Aug. 1991.

\bibitem{Khamisov2024TPS}
O.~O. Khamisov, M.~Ali, T.~Sayfutdinov, Y.~Jiang, V.~Terzija, and P.~Vorobev, ``A novel contingency-aware primary frequency control for power grids with high {CIG}-penetration,'' \emph{IEEE Transactions on Power Systems}, vol.~39, no.~4, pp. 5792--5805, July 2024.

\bibitem{JIANG2024epsr}
Y.~Jiang, H.~Min, and B.~Zhang, ``Oscillations-aware frequency security assessment via efficient worst-case frequency nadir computation,'' \emph{Electric Power Systems Research}, vol. 234, p. 110656, Sept. 2024.

\bibitem{entsoe2017}
``Oscillation event 03.12.2017---system protection and dynamics {WG},'' European Network of Transmission System Operators for Electricity, Tech. Rep., Mar. 2018.

\bibitem{Kosterev1999TPS}
D.~Kosterev, C.~Taylor, and W.~Mittelstadt, ``Model validation for the august 10, 1996 wscc system outage,'' \emph{IEEE Transactions on Power Systems}, vol.~14, no.~3, pp. 967--979, Aug. 1999.

\bibitem{WP2021}
``Power system stability requirements guideline,'' Electricity Networks Corporation Trading as Western Power, Tech. Rep., May 2021.

\bibitem{rogers2000}
G.~Rogers, \emph{{Power System Oscillations}}, 1st~ed.\hskip 1em plus 0.5em minus 0.4em\relax Springer, 2000.

\bibitem{Gautam2009TPS}
D.~Gautam, V.~Vittal, and T.~Harbour, ``Impact of increased penetration of {DFIG}-based wind turbine generators on transient and small signal stability of power systems,'' \emph{IEEE Transactions on Power Systems}, vol.~24, no.~3, pp. 1426--1434, Aug. 2009.

\bibitem{Pagnier2019Plos}
L.~Pagnier and P.~Jacquod, ``Inertia location and slow network modes determine disturbance propagation in large-scale power grids,'' \emph{PLOS ONE}, vol.~14, no.~3, p. 0213550, Mar. 2019.

\bibitem{zhangtseoscillation}
M.~Zhang, Z.~Miao, and L.~Fan, ``Reduced-order analytical models of grid-connected solar photovoltaic systems for low-frequency oscillation analysis,'' \emph{IEEE Transactions on Sustainable Energy}, vol.~12, no.~3, pp. 1662--1671, July 2021.

\bibitem{Majumder2006TPS}
R.~Majumder, B.~Pal, C.~Dufour, and P.~Korba, ``Design and real-time implementation of robust facts controller for damping inter-area oscillation,'' \emph{IEEE Transactions on Power Systems}, vol.~21, no.~2, pp. 809--816, May 2006.

\bibitem{Feng2025TPS}
C.~Feng, L.~Huang, X.~He, Y.~Wang, F.~Dörfler, and C.~Kang, ``Hybrid oscillation damping and inertia management for distributed energy resources,'' \emph{IEEE Transactions on Power Systems}, vol.~40, no.~6, pp. 5041--5056, Nov. 2025.

\bibitem{guo2018cdc}
L.~{Guo}, C.~{Zhao}, and S.~H. {Low}, ``Graph laplacian spectrum and primary frequency regulation,'' in \emph{Proc. of IEEE Conference on Decision and Control}, Dec. 2018, pp. 158--165.

\bibitem{Minl4dc23cluster}
H.~Min and E.~Mallada, ``Learning coherent clusters in weakly-connected network systems,'' in \emph{Proc. of Learning for Dynamics and Control Conference}, June 2023, pp. 1167--1179.

\bibitem{rouco1998eigenvalue}
L.~Rouco, ``Eigenvalue-based methods for analysis and control of power system oscillations,'' in \emph{Proc. of IEE Colloquium on Power System Dynamics Stabilisation (Digest No. 1998/196 and 1998/278)}, Feb. 1998, pp. 3/1--3/6.

\bibitem{Ke2011TPS}
D.~P. Ke, C.~Y. Chung, and Y.~Xue, ``An eigenstructure-based performance index and its application to control design for damping inter-area oscillations in power systems,'' \emph{IEEE Transactions on Power Systems}, vol.~26, no.~4, pp. 2371--2380, Nov. 2011.

\bibitem{pm2019preprint}
F.~{Paganini} and E.~{Mallada}, ``Global analysis of synchronization performance for power systems: Bridging the theory-practice gap,'' \emph{IEEE Transactions on Automatic Control}, vol.~65, no.~7, pp. 3007--3022, July 2020.

\bibitem{Zhao:2013ts}
C.~Zhao, U.~Topcu, N.~Li, and S.~H. Low, ``Power system dynamics as primal-dual algorithm for optimal load control,'' \emph{arXiv preprint:1305.0585}, May 2013.

\bibitem{Purchala2005dc-flow}
K.~{Purchala}, L.~{Meeus}, D.~{Van Dommelen}, and R.~{Belmans}, ``Usefulness of {DC} power flow for active power flow analysis,'' in \emph{Proc. of IEEE Power Engineering Society General Meeting}, June 2005, pp. 454--459.

\bibitem{UCTLbook}
``Continental europe operation handbook: P1—policy 1: Load-frequency control and performance [{C}],'' Union for the Coordination of Transmission of Electricity, Tech. Rep., 2009.

\bibitem{NG2016standard}
``National electricity transmission system security and quality of supply standard,'' National Grid Electricity System Operator, Tech. Rep., Apr. 2019.

\bibitem{Hara2014tac}
S.~Hara, H.~Tanaka, and T.~Iwasaki, ``Stability analysis of systems with generalized frequency variables,'' \emph{IEEE Transactions on Automatic Control}, vol.~59, no.~2, pp. 313--326, Feb. 2014.

\bibitem{Chilali1996tac}
M.~Chilali and P.~Gahinet, ``${H}_\infty$ design with pole placement constraints: an {LMI} approach,'' \emph{IEEE Transactions on Automatic Control}, vol.~41, no.~3, pp. 358--367, Mar. 1996.

\bibitem{Driels1996linear}
M.~Driels, \emph{{Linear Control Systems Engineering}}.\hskip 1em plus 0.5em minus 0.4em\relax McGraw-Hill, 1996.

\bibitem{kundur_power_1994}
P.~Kundur, \emph{{Power System Stability and Control}}.\hskip 1em plus 0.5em minus 0.4em\relax McGraw-Hill, 1994.

\bibitem{Vorobev2019tps}
P.~Vorobev, D.~M. Greenwood, J.~H. Bell, J.~W. Bialek, P.~C. Taylor, and K.~Turitsyn, ``Deadbands, droop, and inertia impact on power system frequency distribution,'' \emph{IEEE Transactions on Power Systems}, vol.~34, no.~4, pp. 3098--3108, July 2019.

\bibitem{chow1992toolbox}
J.~H. Chow and K.~W. Cheung, ``A toolbox for power system dynamics and control engineering education and research,'' \emph{IEEE transactions on Power Systems}, vol.~7, no.~4, pp. 1559--1564, Nov. 1992.

\end{thebibliography}

\end{document}